\documentclass[12pt]{article}
\usepackage[a4paper,top=2cm,bottom=2cm,left=2cm,right=2cm]{geometry}
\usepackage{amsmath}
\usepackage{amsthm}
\usepackage{amssymb}

\usepackage{subfig}
\usepackage{graphicx}
\usepackage{booktabs}
\usepackage{braket}
\usepackage{enumerate}
\usepackage[title]{appendix}
\usepackage[colorlinks=true, allcolors=blue]{hyperref}
\usepackage{pgfplots}
\newtheorem{lemma}{Lemma}
\newtheorem{theorem}[lemma]{Theorem}
\usepackage[ruled]{algorithm2e}
\usepackage[most]{tcolorbox}
\newtcbtheorem{mytcbprob}{Problem}{}{problem}

\title{Deterministic quantum search with adjustable parameters: implementations and  applications}
\author{ Guanzhong Li\thanks{Email: ligzh9@mail2.sysu.edu.cn}~ and Lvzhou Li\thanks{Email: lilvzh@mail.sysu.edu.cn (corresponding author)}\\
\small{{\it Institute of Quantum Computing and Computer Theory,}}\\
\small{{\it School of Computer Science and Engineering,}}\\
\small {{\it  Sun Yat-sen University, Guangzhou 510006, China}}}
\date{\today }

\begin{document}
\maketitle

\begin{abstract}
Grover's algorithm provides a quadratic speedup over classical algorithms to search for marked elements in an unstructured database. The original algorithm is probabilistic, returning a marked element with bounded error.
There are several schemes to achieve the deterministic version, by using the generalized Grover's iteration $G(\alpha,\beta):=S_r(\beta)\, S_o(\alpha)$ composed of phase oracle $S_o(\alpha)$ and phase rotation $S_r(\beta)$. 
However, in all the existing schemes the value range of $\alpha$ and $\beta$ is limited; for instance, in the three early schemes $\alpha$ and $\beta$ are determined by the proportion of marked states $M/N$.
In this paper, we break through this limitation by presenting a search framework with adjustable parameters, which allows $\alpha$ or $\beta$ to be arbitrarily given.
The significance of the framework lies not only in the expansion of mathematical form, but also in its application value, as we present two disparate problems which we are able to  solve deterministically using the proposed framework, whereas previous schemes are ineffective.

\end{abstract}

\section{Introduction}
\subsection{Background}
Grover's algorithm \cite{Grover} is one of the most fundamental algorithms in quantum computing, as it provides a quadratic speedup for the unstructured search problem. 
In this  problem, the user is allowed to query an oracle, which outputs information about whether an element is marked, after receiving the element index. In quantum computing, one may input superposition of different indexes to the oracle, and it will also output a superposition of answers. 
Grover's algorithm cleverly makes use of this characteristic and achieves quadratic speedup compared to classical query-based search algorithms. It also works in the very general context, because the oracle is regarded as a black box so that there need not be any promise on the structure of the database. This black box feature makes it easy to generalize Grover's algorithm to quantum amplitude amplification \cite{QAA}, which has become an important subroutine in designing many other quantum algorithms.

 Due to its great importance, research even just on Grover's algorithm itself has never stopped ever since it was proposed. 
One question is whether Grover's algorithm can be made exact, as the original  Grover's algorithm can only output a marked element with a high probability, but not with certainty. To this end, three methods \cite{QAA,arbi_phase,Long} were proposed, which are based on different ideas but all, for the first step, change the two phase flip operations in the original Grover's algorithm to the following general operations: 
\begin{equation}
    S_o(\alpha)=e^{i\alpha\, \Pi_\mathcal{M}},
\end{equation}
\begin{equation}
    S_r(\beta)=e^{-i\beta\, \ket{\psi_0}\bra{\psi_0}}.
\end{equation}
$S_o(\alpha)$ rotates the phase of the marked states by $e^{i\alpha}$, where $\Pi_\mathcal{M}$ represents orthogonal projector onto the subspace spanned by marked states. $S_r(\beta)$
rotates the phase of the initial state $\ket{\psi_0}$ (an equal-superposition of all indexes) by $e^{-i\beta}$.
Note that when $\alpha=\beta=\pi$, the generalized Grover's operation
\begin{equation}
    G(\alpha,\beta):=S_r(\beta)\, S_o(\alpha)
\end{equation}
recovers to the original Grover's iteration. A quantum circuit which realizes the generalized Grover's operation $G(\alpha,\beta)$ using 
twice the standard oracle $O_f$ is shown in Fig. \ref{fig:circuit}.

\begin{figure}[ht]
\centering
\includegraphics[width=0.8\textwidth]{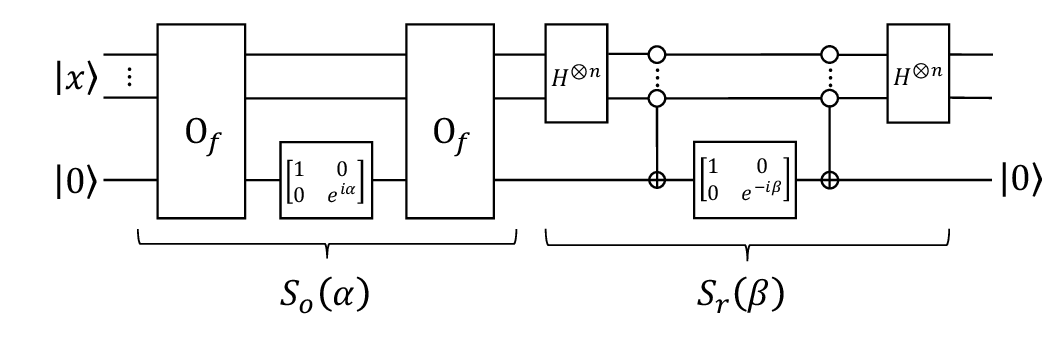}
\caption{\label{fig:circuit} A quantum circuit to realize the generalized Grover's operation $G(\alpha,\beta)=S_r(\beta)\, S_o(\alpha)$. The first $n$ horizontal lines represents working qubits storing $N=2^n$ element indexes, and the last line represents an ancillary qubit initially set to $\ket{0}$. The  standard oracle $O_f$ maps $\ket{x}\ket{b}$ to $\ket{x}\ket{b\oplus f(x)}$, where $\oplus$ is the XOR operation, and the Boolean function $f(x)=1$ iff $x$ represents a marked element index. The vertical line with circles is the multiple controlled-NOT gate which flips the ancillary qubit whenever the first $n$ qubits are all $\ket{0}$. Correctness of the circuit implementing $S_o(\alpha)$ can be easily verified by considering basis state $\ket{x}$ when $f(x)=0$ or $1$ separately. Note that $S_r(\beta)$ is equal to $H^{\otimes n} (I-(1-e^{-i\beta})\ket{0}^{\otimes n} \bra{0}^{\otimes n}) H^{\otimes n}$.  }
\end{figure}

The three methods \cite{QAA,arbi_phase,Long} then choose appropriate parameters $\alpha,\beta,k$ based on different ideas (which we summarize them as `big-small step', `conjugate rotation', and `3D rotation'), so that after $k$ iterations of $G(\alpha,\beta)$ to the initial state $\ket{\psi_0}$, an equal superposition of marked states is produced, and measurement in the computational basis will lead to a marked state with certainty. A summary of them is presented in 
Table \ref{tab:three_exact}.

\begin{table}[ht]
    \centering
    \begin{tabular}{ccc}
         \toprule
         method & procedure & $k$ \\
        \midrule
         big-small step \cite{QAA}
         & $G(\alpha_1,-\beta_1)G^k(\pi,\pi) \ket{\psi_0}$ 
         & $\lfloor k_\mathrm{opt} \rfloor$ \\
         conjugate rotation \cite{arbi_phase}
         & $G^k(\alpha_2,-\beta_2) S_o(u) \ket{\psi_0}$
         & $\lceil k_\mathrm{opt} \rceil$ \\
         3D rotation \cite{Long}
         & $G^k(\alpha_3,-\alpha_3) \ket{\psi_0}$
         & $\lceil k_\mathrm{opt} \rceil$ \\
         \bottomrule
    \end{tabular}
    \caption{Three methods to achieve exact Grover search. Parameter $k_\mathrm{opt}=\pi/2\theta-1/2$, where $\theta=2\arcsin\sqrt{\lambda}$, and $ \lambda=M/N$ is the proportion of marked elements.
    Parameters $(\alpha_1,\beta_1)$ satisfy (are solution to) the complex valued equation ``$(-\cos\theta + i\cot\frac{\alpha_1}{2}) \cot(2k+1)\frac{\theta}{2} = e^{i\beta_1} \sin\theta$''.
    Parameters $(\alpha_2,\beta_2)$ satisfy the equations ``$\sin\frac{\alpha_2}{2} \sin\theta = \sin\frac{\pi-\theta}{2k}$'' and ``$\tan\frac{\beta_2}{2} = \tan\frac{\alpha_2}{2} \cos\theta$''. Parameter $u=\frac{\pi-\beta_2}{2}$.
    Parameter $\alpha_3$ satisfies the equation ``$\sin\frac{\pi}{4k+2} = \sin\frac{\alpha_3}{2} \sin\frac{\theta}{2}$''.}
    \label{tab:three_exact}
\end{table}

Compared to the original Grover's algorithm, the three methods produce a marked state with certainty while maintaining the quadratic speedup (note that $k_\text{opt}\approx \frac{\pi}{4}\sqrt{\frac{M}{N}}$ ), but they require knowledge of the ratio $\lambda=M/N$. In fact, the ratio $\lambda$ is {\it necessary} to design quantum exact search algorithm, as it has been proved in Ref. \cite{polynomial} that any quantum algorithm for exact search without knowing the number $M$ of marked elements requires $N$ queries (by reducing the computation of Boolean function $OR_N$ to exact search, and proving the former requires $N$ queries to achieve certainty).

\subsection{Contributions}
As can be seen from Table \ref{tab:three_exact}, the parameters $\alpha,\beta$ in $G(\alpha,\beta)=S_r(\beta)\, S_o(\alpha)$ are fixed once the ratio $\lambda$ is known. One question now naturally arises as follows.

\noindent\textbf{The question:} \textit {Is it possible to achieve exact search when one of the parameters $\alpha,\beta$ in $G(\alpha,\beta)$ is arbitrarily given? Or to be more precise, when the angle $\alpha$ in phase oracle $S_o(\alpha)$, or the angle $\beta$ in phase rotation  $S_r(\beta)$ is arbitrarily given? }

In this paper, we answer the above question affirmatively. Specifically, we have the following two theorems for the two cases in which $\alpha$ or $\beta$ is arbitrarily given:
\begin{theorem}\label{thm:main}
Given the phase oracle $S_o(\alpha)$ with an arbitrary angle $\alpha\in (0,2\pi)$,  whenever $k > k_{\mathrm{lower}}$,  there always exits a pair of parameters $(\beta_1,\beta_2)$, such that applying $F_{xr,\alpha} := G(\alpha,\beta_2) G(\alpha,\beta_1)$ to $\ket{\psi_0}$ for $k$ times will produce an equal superposition of marked states. 
The lower bound of  $k$ is
\begin{equation}
    k_{\mathrm{lower}} = \frac{\pi}{ \Big| 4\arcsin(\sqrt{\lambda}\sin\frac{\alpha}{2}) \mod [-\frac{\pi}{2},\frac{\pi}{2}] \Big| }\in O(1/\sqrt{\lambda}), \label{k_lower_intro}
\end{equation} where $\lambda\in(0,1)$ is the proportion of marked elements,  and the notation ``$x \mod [-\frac{\pi}{2},\frac{\pi}{2}]$'' means to add $x$ with an appropriate integer multiples $l$ of $\pi$, such that $x+l\pi\in [-\frac{\pi}{2},\frac{\pi}{2}]$. 
\end{theorem}

\begin{theorem}\label{thm:main2}
Given the phase rotation   $S_r(\beta)$ with an arbitrary angle $\beta\in (0,2\pi)$, whenever $k > k_{\mathrm{lower}}$, there always exits a pair of parameters $(\alpha_1,\alpha_2)$, such that applying $F_{xr,\beta} := G(\alpha_1,\beta) G(\alpha_2,\beta)$ to $\ket{\psi_0}$ for $k$ times will produce an equal superposition of marked states. 
The lower bound of  $k$ is
\begin{equation}
    k_{\mathrm{lower}} = \frac{\pi}{ \Big| 4\arcsin(\sqrt{\lambda}\sin\frac{\beta}{2}) \mod [-\frac{\pi}{2},\frac{\pi}{2}] \Big| }\in O(1/\sqrt{\lambda}). \label{k_lower2_intro}
\end{equation}
\end{theorem}

In Section \ref{sec:app}, two applications of the above results will be presented. Here we have a look at the proof outline.

{\it Proof outline of Theorem~\ref{thm:main} and \ref{thm:main2}:}
The proof is divided into two parts: 

1.  To derive the explicit equations that parameters $(k,\beta_1,\beta_2)$ or $(k,\alpha_1,\alpha_2)$ need to satisfy. This is done in Section~\ref{sec:d2p} with the final equations shown by Eqs.~\eqref{eq:intro_fxr_alpha_real},\eqref{eq:intro_fxr_alpha_im} for $(k,\beta_1,\beta_2)$; and Eqs.~\eqref{eq:intro_fxr_beta_real},\eqref{eq:intro_fxr_beta_im} for $(k,\alpha_1,\alpha_2)$. We first determine the effect of $F_{xr}=G(\alpha_2,\beta_2) G(\alpha_1,\beta_1)$ in its invariant subspace spanned by $\ket{R}$ and $\ket{T}$, which are the equal-superposition of unmarked and marked states respectively.
In fact, $F_{xr}$ can be seen as a rotation $R_{\vec{n}}(\phi)$ composing of four rotations $S_r(\beta_2) S_o(\alpha_2) S_r(\beta_1) S_o(\alpha_1)$ on the Bloch sphere of basis $\{\ket{R},\ket{T}\}$, so we can then calculate the rotation parameters $\vec{n}$ and $\phi$ of $F_{xr,\alpha}$ and $F_{xr,\beta}$ (Lemma~\ref{lem:decomp_fxr}).
The Bloch sphere interpretation of single-qubit operations provides us with the tool to deal with the exponentiation $F_{xr}^k$, because of the equation $R^k_{\vec{n}}(\phi) = R_{\vec{n}}(k\phi)$.
Then by considering the real and imaginary part of $0 = \bra{R} F_{xr,\alpha}^k \ket{\psi_0}$ or $0 = \bra{R} F_{xr,\beta}^k \ket{\psi_0}$, we obtain the equations that parameters $(k,\beta_1,\beta_2)$ or $(k,\alpha_1,\alpha_2)$ need to satisfy. 

2. To prove that whenever $k$ is greater than $k_{\mathrm{lower}}$, there exists a solution $(\beta_1,\beta_2)$ or $(\alpha_1,\alpha_2)$ to their corresponding equations.
This part is quite technical and is shown in Section~\ref{sec:iterations k}, where we make use of 3D geometry (Lemma~\ref{thm:phi zero}) and the intermediate value theorem of continuous function on closed intervals.
We only show the proof for the case where $\alpha$ is fixed, since the proof for $\beta$ being fixed is almost the same, but we point out some differences in Section~\ref{subsec:differ}.


\textbf{Remark:} We call the procedure in Theorem~\ref{thm:main} and \ref{thm:main2} {\it the fixed-axis-rotation (FXR) method}, because the exponentiation $F_{xr}^k$ can be seen as rotating about a fixed axis for $k$ times on the Bloch sphere.
The peculiar notation in the denominator of Eq.~\eqref{k_lower_intro} only takes effect when $\lambda$ is large, because when $\lambda$ is small, $4\arcsin(\sqrt{\lambda}\sin\frac{\alpha}{2})$ will always be in $[0,\frac{\pi}{2}]$, and $k_{\mathrm{lower}}$ will be of order $O(1/\sqrt{\lambda})$, maintaining the quadratic speedup. The same goes for Eq.~\eqref{k_lower2_intro}.

\subsection{Related work and paper structure}
A special case of Theorem~\ref{thm:main} with $\alpha=\pi$  was considered by Roy et al. in \cite{d2p}, where  the deterministic two-parameter (D2p) protocol was proposed. The D2p protocol obtains zero failure rate by applying $G(\pi,\beta_1)$ and $G(\pi,\beta_2)$, alternatively, to the initial state $\ket{\psi_0}$ until $k\geq\lceil k_{\mathrm{opt}} \rceil$ oracle queries are made. When $k$ is even, the last iteration is $G(\pi,\beta_2)$; when $k$ is odd, the last iteration is $G(\pi,\beta_1)$ and the equations  which parameters $(\beta_1,\beta_2)$ need to satisfy are more complicated.
Compared to the D2p protocol, our FXR method iterates $F_{xr}$ as one piece for $k$ times, regardless of the parity of $k$, therefore only one system of equations needs to be solved to obtain parameters $(\beta_1,\beta_2)$.
Moreover, they claimed that the equations can always be solved when $k\geqslant k_{\text{opt}}$ and $\lambda\leqslant 1/4$, but {\it lack} rigorous proof.

 The rest of this paper is organized as follows. 
We first give two applications of our results in Section~\ref{sec:app}. The first uses Theorem~\ref{thm:main} and the other uses Theorem~\ref{thm:main2}.
In Section~\ref{sec:d2p} we describe the intuition behind our FXR method,
and derive the equations that parameters $(k,\beta_1,\beta_2)$ or $(k,\alpha_1,\alpha_2)$ need to satisfy. In Section~\ref{sec:iterations k} we prove the existence of solution $(\beta_1,\beta_2)$ or $(\alpha_1,\alpha_2)$ to the corresponding equations whenever $k$ is greater than $k_{\mathrm{lower}}$. Conclusion and future work are discussed in Section~\ref{sec:conclusion}.

\section{Applications}\label{sec:app}
In this section, we present two applications of the search framework with adjustable parameters. The first uses Theorem~\ref{thm:main}, and the second uses Theorem~\ref{thm:main2}.
It can be seen especially from the second application that the search framework can be deployed in designing quantum exact algorithms no longer restricted to Grover search problem, and the key step is to reduce the problem under consideration to the problem of transforming initial state to target state in a two-dimensional subspace with some obtained $G(\alpha,\beta)$.
We believe that the two applications shown below are only tip of the iceberg demonstrating the usefulness of the search framework, and that more interesting ones will be found in the future.

\subsection{Identifying secret string exactly with Hamming distance oracle}
Hunziker and Meyer \cite{highly} considered the problem of identifying a secret string $s\in[k]^n$, by querying the Hamming distance oracle $h_s(x)=dist(x,s) \,\text{mod}\,r$, where $r=\max\{2,6-k\}$, and $dist(x,s)$ is the generalized Hamming distance between the query string $x$ and $s$, i.e., the number of components at which they differ.  When $k\in\{2,3,4\}$,   Algorithm B  proposed in \cite{highly} requires only one query to identify $s$ exactly. However, the case of $k>4$ is more complicated, and   we recall the main idea proposed in \cite{highly}  for this case in the following, where the second step has a flaw and can be amended by using Theorem~\ref{thm:main}.
\begin{enumerate}
    \item [(i)]  Algorithm C  proposed in \cite{highly}  can be viewed as $n$ synchronous Grover search on each position of $s$. In the $i$-th position, a Grover search with $O(\sqrt{k})$ iterations is performed for finding one marked item out of $k$ items, and thus $s_i$ is identified correctly with probability at worst $1-\frac{1}{k}$ \cite{boyer1998tight}. As a result, the probability of identifying $s$ is at least $(1-\frac{1}{k})^n$  which is bounded above $ \frac{1}{2}+\epsilon$ when $n < -k\ln(\frac{1}{2} + \epsilon)$.

\item [(ii)] The Grover search on each position can be adjusted to succeed with probability $1$ by using one of the three methods in Table~\ref{tab:three_exact}.  

\item[(iii)] Thus, the condition $n < -k\ln(\frac{1}{2} + \epsilon)$ can be removed, which leads to a quantum algorithm identifying $s$ with certainty and consuming $O(\sqrt{k})$ queries.

\end{enumerate}

\begin{algorithm}[htb]
    \SetKwInput{Runtime}{Runtime}
     \SetKwInOut{KWProcedure}{Procedure}
    \caption{Algorithm C in \cite{highly} }
    \label{algorithm:meyer}
    \LinesNumbered
    \KwIn {An oracle $O_{h_s}$ for $s\in [k]^n$ such that $O_{h_s} | x \rangle | b \rangle = | x \rangle | b \oplus h_s(x) \rangle$.}
    \KwOut {The secret string $s$.}
    \Runtime {$O(\sqrt{k})$ queries to $O_{h_s}$.  Succeeds with probability at least $\frac{1}{2} + \epsilon $ when $n < -k\ln(\frac{1}{2} + \epsilon)$.}
    
    \KWProcedure{}
    Initial the state to $ | 0 \rangle ^{\otimes n} | 0 \rangle \in (C_k)^{\otimes n}\otimes C_2$.
    
    Apply the unitary transformation $QFT_{k}^{\otimes n} \otimes (HX)$
    
    \For{ $i = 1 :\lfloor \frac{1}{2}(\pi / (2\arcsin(\frac{1}{\sqrt{k}})) - 1)\rceil $}{
        apply the oracle $O_{h_s}$.
        
        apply the unitary transformation $(QFT_k (I - 2|0 \rangle \langle 0 |)QFT_k^{\dagger})^{\otimes n } \otimes I$
    }
    Measure the first $n$ registers.
\end{algorithm}

However, as pointed out in \cite{mastermind}, the second item above is not true, that is, the algorithm  cannot be adjusted to succeed with probability $1$ by using the three methods in Table~\ref{tab:three_exact}. We explain the reason below. 

First note that when $k>4$, we have $h_{s}(x) = ({n - \sum_{i = 1}^n \delta_{s_ix_i}}) \bmod 2$, and the quantum oracle works as  $O_{h_s} | x \rangle | b \rangle = | x \rangle | b \oplus h_s(x) \rangle $ where $\oplus$ denotes the bitwise XOR. More specifically,  the oracle  $O_{h_a}$  works as follows:
\begin{align}
    O_{h_s} (|x \rangle \otimes \ket{-})
    \label{equation:meyer1}& = | x \rangle \otimes \frac{1}{\sqrt{2}}(| 0 \oplus h_s(x) \rangle - | 1 \oplus h_s(x) \rangle)\\
    \label{equation:meyer2}& = (-1)^{h_s(x)} | x \rangle \otimes \ket{-}\\ 
    \label{equation:meyer3}
    & = (-1)^{(n - \sum_{i = 1}^n \delta_{s_ix_i}) \bmod 2} | x \rangle \otimes \ket{-}\\ 
    \label{equation:meyer3.5}
    & = (-1)^{n - \sum_{i = 1}^n \delta_{s_ix_i}} | x \rangle \otimes\ket{-}\\ 
    \label{equation:meyer4}
    & = (-1)^n\mathop{\bigotimes}_{i = 1}^{n} (-1)^{\delta_{s_ix_i}}|x_i\rangle \otimes \ket{-},
\end{align}
where $\ket{-}$ denotes the state $\frac{1}{\sqrt{2}}(\ket{0}-\ket{1})$. 
The trick employed there is phase pick-back as shown in Eq.~\eqref{equation:meyer1} and Eq.~\eqref{equation:meyer2}, adding a fixed phase $-1$ to $| x_i \rangle$ when $x_i$ is equal to $s_i$. It should be pointed out that  Eq.~\eqref{equation:meyer3.5} holds as  $
(-1)^{ l} = (-1)^{ l\bmod 2}$ for any $0\leq l\leq n$, but it will not hold if we replace $-1$ with $ e^{i\phi}$ for general $\phi$, since $ e^{i\phi l} = e^{i\phi l\bmod 2}$ no longer holds.
On the other hand, the three methods in Table~\ref{tab:three_exact} requires the phase oracle $S_o(\alpha)$ to have arbitrary user-controllable angle $\alpha$. That is why the step in the second item above cannot be adjusted to succeed with probability $1$ by using the three methods in Table~\ref{tab:three_exact}.

Nevertheless, we can use the FXR method instead, by setting the phase oracle's angle $\alpha=\pi$ and $\lambda=1/k$ in Theorem~\ref{thm:main}. We will then identify the secret string $s$ by measuring the first $n$ registers, and what's more, the query complexity remains $O(\sqrt{k})$ by Eq.~\eqref{k_lower_intro}.

\subsection{Solving element distinctness promise problem deterministically}
The element distinctness problem is to determine whether a string $x=(x_1,x_2,\ldots,x_N)\in [M]^N$ of $N$ elements, where $[M]=\{1,2,\ldots,M\}$ and $M\geq N$, contains two elements of the same value.
More precisely, we are given a black box (oracle) $O_x$ that when queried index $i$ of the unknown string $x=(x_1,x_2,\ldots,x_N)$, it will output its value $x_i$, and the task now is to determine whether the string $x$ contains two equal items $x_{i}=x_{j}$ (called colliding pair), with as few queries to the index oracle $O_x$ as possible.

The problem is well-known both in classical and quantum computation. The classical bounded-error query complexity (a.k.a decision tree complexity, see \cite{decision_tree_survey} for more information) is shown to be $\Theta(N)$ by a trivial reduction from the OR-problem \cite{quantum_alg1.2}. In comparison, Ambainis \cite{quantum_alg2.1,quantum_alg2.2} proposed a $O(N^{2/3})$-query quantum algorithm with bounded-error, and it is optimal: the bounded-error quantum query complexity of element distinctness problem has previously been shown to be $\Omega(N^{2/3})$ \cite{quantum_lower1,quantum_lower2,quantum_lower3} by the polynomial method \cite{polynomial}.
The quantum algorithm proposed by Ambainis \cite{quantum_alg2.1,quantum_alg2.2} is developed in two steps:
\begin{enumerate}[(i)]
    \item Ambainis first designs a $O(N^{2/3})$-query quantum algorithm $\mathcal{A}$ which solves the simple case ({\it promise problem}) in which string $x$ contains at most one colliding pair, by using quantum walk search on Johnson graph. 
    \item He then reduces the general case to this simple case, by sampling a sequence of exponentially shrinking subsets to run $\mathcal{A}$. 
    By showing the probability that the sequence has a subset containing at most one colliding pair is greater than $1/2$, and the overall query complexity remains the same order, Ambainis successfully solves the element distinctness problem with query complexity $O(N^{2/3})$.
\end{enumerate}
This two stage process shows the {promise problem} that algorithm $\mathcal{A}$ solves is important. We formalize it in the following box.

\begin{mytcbprob*}{the element distinctness promise problem}
\textbf{input:} index oracle $O_x$ hiding the unknown string $x=(x_1,x_2,\ldots,x_N)\in [M]^N$.

\textbf{promise:} $x$ contains exactly one colliding pair $x_i=x_j$ with $i\neq j$; or $x_i$s are distinct for $i\in[N]$.

\textbf{output:} indexes $(i,j)$ of the colliding pair $x_i=x_j$; or ``all distinct”.
\end{mytcbprob*}

However, $\mathcal{A}$ is a one-sided error algorithm: if the string $x$ contains a colliding pair $x_i=x_j$, then $\mathcal{A}$ might asserts that $x$ is ``all distinct'', making a mistake; but if elements in $x$ are all distinct, then $\mathcal{A}$ won't err.
This motivates us to develop in \cite{eedp} a quantum exact algorithm that solves the element distinctness promise problem with $100\%$ success and has the same query complexity $\Theta(N^{2/3})$.
The design of our quantum exact algorithm can be roughly divided into the following three steps:
\begin{enumerate}
    \item We first reduce the problem to quantum walk search on quasi-Johnson graph $G$, based on Portugal's staggered quantum walk model \cite{Portugal_2018,Portugal_book}, but we introduce arbitrary phase angles in the two local diffusion operators, in order to achieve certainty of success in the next two steps.
    \item We then reduce the quantum walk search on $G$ to exact Grover search problem with an arbitrarily given phase rotation $S_r(\beta)$. This step is where our novelty lies: we use Jordan's lemma about common invariant subspaces of two reflections and some intuition from rotation on the Bloch sphere, to reduce the quantum walk's five-dimensional invariant subspace further to a two-dimensional  space, enabling us to use the search framework.
    \item We use Theorem~\ref{thm:main2} to transform the initial state to the target state in the reduced two-dimensional space obtained  in step 2.
\end{enumerate}

We have to use in the last step Theorem~\ref{thm:main2} instead of the other three methods in Table~\ref{tab:three_exact}, because in the phase rotation operator $S_r(\beta)$ that we've obtained from step 2, the angle $\beta$ is a solution to some equations and is therefore arbitrarily given.

\section{The fixed-axis-rotation method}\label{sec:d2p}
In this section, we first reformulate the problem of searching an unstructured database based on phase oracles, and then describe our FXR method to achieve quantum exact search with one of the parameters $\alpha,\beta$ in $G(\alpha,\beta)$ being arbitrarily given. Finally, we derive the under-determined system of equations that parameters $(\beta_1,\beta_2,k)$ or $(\alpha_1,\alpha_2,k)$ need to satisfy respectively, when $\alpha$ or $\beta$ is arbitrarily given.

\subsection{Preliminaries}

In the problem of searching an unstructured database, suppose that there are $N$ elements in the database, $M$ of which are marked, so the ratio of marked elements $\lambda=M/N$ is known in advance. Suppose the marked elements' indexes are represented by the computational basis $\{\ket{t_j}\mid j=1,2,\ldots,M\}:=\mathcal{M}$, and the unmarked ones by $\{\ket{r_j}\mid j=1,2,\ldots,N-M\}$. Then the phase oracle $S_o(\alpha)$ is 
\begin{equation}
    S_o(\alpha) = e^{i\alpha} \sum_{j=1}^{M} \ket{t_j}\bra{t_j} + \sum_{j=1}^{N-M} \ket{r_j}\bra{r_j} = e^{i\alpha\, \Pi_\mathcal{M}},
\end{equation}
which rotates the phase of the marked states $\ket{t_j}$ by an angle $\alpha$ and keeps the unmarked states $\ket{r_j}$ unchanged.
We now want to design an quantum exact search algorithm which outputs one of the marked states with certainty, by querying the oracle $S_o(\alpha)$ or applying $S_r(\beta)$ as few as possible.

As usual, we starts with preparing an equal-superposition state
\begin{equation}
\ket{\psi_0} = \frac{1}{\sqrt{N}} \sum_{j=0}^{N-1} \ket{j},
\end{equation}
which can be produced by, for example, applying $H^{\otimes n}$ to $\ket{0}^{\otimes n}$ (when $2^n=N$). If we denote by $\ket{R}, \ket{T}$ the equal-superposition of unmarked and marked states respectively:
\begin{equation}
    \ket{R} := \frac{1}{\sqrt{N-M}} \sum_{j=1}^{N-M} \ket{r_j}, \quad \ket{T} := \frac{1}{\sqrt{M}} \sum_{j=1}^{M} \ket{t_j}.
\end{equation}
Then $\ket{\psi_0}$ lies in the two-dimensional subspace $\mathcal{H}_0$ spanned by orthonormal vectors $\ket{R}$ and $\ket{T}$, i.e.
\begin{equation}
    \mathcal{H}_0 = \text{span} \{ \ket{R},\ket{T} \},
\end{equation}
and can be expressed as $\ket{\psi_0} = \sqrt{1-\lambda}\ket{R}+\sqrt{\lambda}\ket{T}$. 
It is easy to see that  subspace $\mathcal{H}_0$ is invariant under the phase oracle $S_o(\alpha)$, which has the following matrix representation in its basis $\{ \ket{R}=[1,0]^T \ , \ket{T}=[0,1]^T \}$:
\begin{equation}\label{S_o}
S_o(\alpha) =
\begin{bmatrix}
1 & 0\\
0 & e^{i\alpha}
\end{bmatrix}.
\end{equation}

The phase rotation $S_r(\beta)$ is defined by
\begin{equation}
S_r(\beta)
=I-(1-e^{-i\beta})\ket{\psi_0}\bra{\psi_0} = e^{-i\beta\, \ket{\psi_0}\bra{\psi_0}}, \label{S r def}
\end{equation}
which also maps subspace $\mathcal{H}_0$ to itself, and has the following matrix representation in the same basis:
\begin{equation}\label{s_r}
S_r(\beta) =
\begin{bmatrix}
1-(1-e^{-i\beta})(1-\lambda) & -(1-e^{-i\beta})\sqrt{\lambda(1-\lambda)} \\
-(1-e^{-i\beta})\sqrt{\lambda(1-\lambda)} & 1-(1-e^{-i\beta})\lambda
\end{bmatrix}.
\end{equation}
In fact, $S_r(\beta)=I-(1-e^{-i\beta})\left( \sqrt{1-\lambda}\ket{R}+\sqrt{\lambda}\ket{T} \right) \left(\sqrt{1-\lambda}\bra{R}+\sqrt{\lambda}\bra{T}\right)$. Thus, $\bra{R}S_r(\beta)\ket{R} = 1-(1-e^{-i\beta})(1-\lambda)$,   $\bra{R}S_r(\beta)\ket{T} = \bra{T}S_r(\beta)\ket{R} = -(1-e^{-i\beta})\sqrt{\lambda(1-\lambda)}$, and 
$\bra{T}S_r(\beta)\ket{T} = 1-(1-e^{-i\beta})\lambda$, giving the matrix representation in Eq.~\eqref{s_r}.

With subspace $\mathcal{H}_0$ in mind, designing an quantum exact search algorithm now boils down to transforming the initial state $\ket{\psi_0}=[\sqrt{1-\lambda},\, \sqrt{\lambda}]^T$ to the target state $\ket{T}=[0,1]^T$, by multiplying a series of parameterized matrices $S_r(\beta)$ and $S_o(\alpha)$ for as few times as possible.

\subsection{Overview of the FXR method}
The {geometry intuition} underlying our FXR method is that one needs only two rotations of noncollinear axes to span the full $SU(2)$ space. In the case when $\alpha$ is fixed, the two rotations are $G(\alpha,\beta_1)$ and $G(\alpha,\beta_2)$. And in the case when $\beta$ is fixed, the two rotations are $G(\alpha_1,\beta)$ and $G(\alpha_2,\beta)$.
So by applying
\begin{equation}\label{eq:def_fxr}
    F_{xr,\alpha} = G(\alpha,\beta_2) G(\alpha,\beta_1), \quad \text{or} \quad F_{xr,\beta} = G(\alpha_2,\beta) G(\alpha_2,\beta)
\end{equation}
to initial state $\ket{\psi_0}$ for a number of $k$ times, where parameters $(\beta_1,\beta_2,k)$ or $(\alpha_1,\alpha_2,k)$ are to be determined later in Section~\ref{subsec:equations}, we hope to get $\ket{T}$ (possibly up to some global phase).

We have to iterate $F_{xr,*}$ for above a number of times for the following reason. Note that $F_{\text{xr},*}$ will be a fixed rotation once the pair of parameters $(\alpha_1,\alpha_2)$ or $(\beta_1,\beta_2)$ therein is fixed, so in order to transform $\ket{\psi_0}$ to $\ket{T}$ by iterating this fixed rotation, we must require the direction of the rotation to be correct: from $\ket{\psi_0}$ straight toward $\ket{T}$. 
This is actually a constraint that we impose on the rotation axis of $F_{\text{xr},*}$, and will in turn leads to a constraint relation [specifically, Eq.~\eqref{im1}] on the pair of variable parameters, as explained by step 2 in the proof of Lemma~\ref{thm:phi zero}.
However, this constraint relation also results in the range of the rotation angle $\phi$ of $F_{\text{xr},*}$ to be small, compared to the spherical distance between the initial state $\ket{\psi_0}$ and the target state $\ket{T}$. Therefore, we have to rotate for a number of times in order to reach $\ket{T}$ from $\ket{\psi_0}$.
The reason explained just now is only a geometric intuition, for more rigorous proof of the feasibility of our FXR method when $k>k_{\text{lower}}$, please refer to Section~\ref{sec:iterations k}.

\subsection{Equations that parameters need to satisfy}\label{subsec:equations}
In a nutshell, the equation that parameters $(\beta_1,\beta_2,k)$ or $(\alpha_1,\alpha_2,k)$ need to satisfy is the following:
\begin{equation}
    0 = \bra{R} F_{xr,\alpha}^k \ket{\psi_0},
    \quad \text{or}\quad 0 = \bra{R} F_{xr,\beta}^k \ket{\psi_0} \label{first_eq}.
\end{equation}
Because in the two-dimensional subspace $\mathcal{H}_0$, a state is equal to $\ket{T}$ up to some global phase if and only if it is orthogonal to $\ket{R}$.
What now stands in our way to solve this equation is the exponentiation of matrix, i.e. $F_{xr,*}^k$.
Luckily, the Bloch sphere representation of single-qubit states and operations provides the tool we need. We summarize some facts about it in Lemma~\ref{lem:Bloch}, and one may find its proof in the standard textbook \cite{nielsen_chuang_2010}.

\begin{lemma}\label{lem:Bloch}
Some facts about single-qubit states and operations in the Bloch sphere representation are as  follows.
\begin{enumerate}[(1)]
    \item Any single-qubit state $\ket{\psi}$ (unit vector in the two-dimensional Hilbert space) can be written as
    \begin{equation}
        \ket{\psi}=\cos\frac{\theta}{2}\ket{0} +e^{i\phi}\sin\frac{\theta}{2}\ket{1},
    \end{equation}
    up to some irrelevant global phase $e^{i\varphi}$ which does not affect measurement outcome, so $\ket{\psi}$ can be mapped to a unit vector $[\cos\phi \sin\theta, \sin\phi \sin\theta, \cos\theta]^T$ on the unit sphere of $\mathbb{R}^3$ (the Bloch sphere).
    \item Any single-qubit operation $U$ (two-dimensional unitary matrix) can be decomposed to 
    \begin{equation}
        U=e^{i\varphi} R_{\vec{n}}(\phi), \label{2_decomp}
    \end{equation}
    where
    \begin{align}
        R_{\vec{n}}(\phi)
        &=\cos\frac{\phi}{2} \cdot I -i\sin\frac{\phi}{2} \cdot(n_xX +n_yY +n_zZ) \label{eq:Rn_phi_pauli}\\
        &=\begin{bmatrix}
            c_\phi  -i s_\phi n_z & -is_\phi n_x -s_\phi n_y \label{R n matrix}\\
            -is_\phi n_x +s_\phi n_y & c_\phi  +is_\phi n_z
        \end{bmatrix}
    \end{align}
    represents a rotation by an angle $\phi$ about the axis $\vec{n}=(n_x,n_y,n_z)$ of the Bloch sphere. We use here (and in the sequel) the following \textbf{shorthand} for trigonometric functions:
    \begin{equation}
        s_{\phi} \equiv \sin\frac{\phi}{2}, \quad c_{\phi} \equiv \cos\frac{\phi}{2}. \tag{shorthand 1} \label{eq:trigono}
    \end{equation}
    Note that rotation $R_{\vec{n}}(\phi)$ is determined by parameters $c_\phi$ and $s_\phi \vec{n}$.
    
    \item The exponentiation of $R_{\vec{n}}(\phi)$ has the following nice expression:
    \begin{equation}
        R^k_{\vec{n}}(\phi) = R_{\vec{n}}(k\phi). \label{exponential}
    \end{equation}
    It means that rotation about a fixed axis by angle $\phi$ for $k$ times, is equal to a rotation about the same axis by angle $k\phi$.
    \item The composition of two arbitrary rotation, i.e. $R_{\vec{n}_2}(\phi_2) R_{\vec{n}_1}(\phi_1)$, is also a rotation $R_{\vec{n}}(\phi)$ with parameters given by \begin{align}
        c_\phi &= c_{\phi_1} c_{\phi_2} -(s_{\phi_1} \vec{n}_1)\cdot(s_{\phi_2} \vec{n}_2), \ \text{and} \\
        s_\phi\vec{n} &= c_{\phi_2} (s_{\phi_1} \vec{n}_1) +c_{\phi_1} (s_{\phi_2} \vec{n}_2) +(s_{\phi_2} \vec{n}_2)\times(s_{\phi_1} \vec{n}_1).
    \end{align}
\end{enumerate}
\end{lemma}

We first introduce some notations about the initial state $\ket{\psi_0}$. Let 
\begin{equation}\label{eq:lambda_trigono}
\sqrt{\lambda}=\sin\frac{\theta}{2}=s_\theta, \quad \sqrt{1-\lambda}=\cos\frac{\theta}{2}=c_\theta.
\end{equation}
Then the initial state can be expressed as $\ket{\psi_0} = \sqrt{1-\lambda}\ket{R}+\sqrt{\lambda}\ket{T} = c_\theta \ket{0}+s_\theta \ket{1}$. It is, by Lemma~\ref{lem:Bloch}.(1), the unit vector $[\sin\theta, 0, \cos\theta]^T$ on the Bloch sphere (with $\ket{R}$ and $\ket{T}$ being the north and south pole respectively), and lies in the $xz$-plane forming an angle of $\theta$ with the positive $z$-axis.

As can be seen from Lemma~\ref{lem:Bloch}.(3), Eq.~\eqref{exponential} provides us with the tool to deal with the exponentiation $F_{xr,*}^k$ in Eq.~\eqref{first_eq}. To use this tool, we should decompose $F_{xr,*}$ into the form of Eq.~\eqref{2_decomp}. This is done in Lemma~\ref{lem:decomp_fxr}.

\begin{lemma}\label{lem:decomp_fxr}
The effect of $F_{xr,\alpha}$ on the Bloch sphere is a rotation $R_{\vec{n}}(\phi)$ whose parameters $c_\phi, s_\phi \vec{n}$ are given by Eqs.~\eqref{eq:fxr_alpha_phi}, \eqref{eq:fxr_alpha_nx}, \eqref{eq:fxr_alpha_ny}, \eqref{eq:fxr_alpha_nz} in Appendix~\ref{app:Fxr}. For $F_{xr,\beta}$, its rotation parameters are given by Eqs.~\eqref{eq:fxr_beta_phi}, \eqref{eq:fxr_beta_nx}, \eqref{eq:fxr_beta_ny}, \eqref{eq:fxr_beta_nz}. We leave out the explicit expression here, because they are only intermediate results to deal with Eq.~\eqref{first_eq}.
\end{lemma}

\begin{proof}
The proof is divided into three steps, with detailed calculation presented in Appendix~\ref{app:Fxr}:

1. We first decompose $S_o(\alpha)$ and $S_r(\beta)$ into the form of Eq.~\eqref{2_decomp}. This is shown in the following Lemma~\ref{lem:S_o decomp}. Its proof is presented in Appendix~\ref{proof:S_o}, which may also provide some idea on how to decompose a two-dimensional unitary matrix into the form of Eq.~\eqref{2_decomp}.

\begin{lemma}\label{lem:S_o decomp}
The two-dimensional unitary matrices $S_o(\alpha)$ and $S_r(\beta)$, i.e.\ Eq.~\eqref{S_o} and Eq.~\eqref{s_r}, have the following decomposition in the form of Eq.~\eqref{2_decomp}:
\begin{align}
S_o(\alpha)
&= e^{i\alpha/2}\left( c_\alpha I -i\, s_\alpha  Z \right)
= e^{i\alpha/2}\ R_{R}(\alpha), \\
S_r(\beta)
&=e^{-i\beta/2} \left( c_\beta I -i\, s_\beta(s_{2\theta} X +c_{2\theta} Z) \right)
= e^{-i\beta/2}\ R_{\psi_0}(\beta),
\end{align}
which in the Bloch sphere represents a rotation about the positive $z$-axis (i.e.\ $\ket{R}$) by an angle of $\alpha$, and a rotation about the axis $[\sin\theta, 0, \cos\theta]^T$ (i.e.\ $\ket{\psi_0}$) by an angle of $\beta$, respectively.
\end{lemma}

2. We then determine the effect of $G(\alpha,\beta)=S_r(\beta)S_o(\alpha)$ on the Bloch sphere. As global phase is irrelevant in the measurement outcome, we only need to consider $R_{\psi_0}(\beta)$ and $R_{R}(\alpha)$. By direct calculation and using some properties of Pauli matrices: $XZ=iY$ and $Z^2=I$, it is easy to obtain the parameters [shown in Eqs.~\eqref{eq:G_alpha_angle}, \eqref{eq:G_alpha_axis}] of the following rotation:
\begin{equation}
    G'(\alpha,\beta) := R_{\psi_0}(\beta) R_{R}(\alpha).
\end{equation}

3. Finally, we determine the effects of $F_{xr,\alpha}=G(\alpha,\beta_2) G(\alpha,\beta_1)$ and $F_{xr,\beta}=G(\alpha_2,\beta) G(\alpha_1,\beta)$, or equivalently, the parameters of the following rotations
\begin{align}
F'_{xr,\alpha}=G'(\alpha,\beta_2) G'(\alpha,\beta_1), \quad F'_{xr,\beta}=G'(\alpha_2,\beta) G'(\alpha_1,\beta).
\end{align}
Using Lemma~\ref{lem:Bloch}.(4) and with some calculation (see Appendix~\ref{app:Fxr}), we find that the parameters of the above two rotations are as claimed in Lemma~\ref{lem:decomp_fxr}.
\end{proof}

We now consider the equation
\begin{equation}
    0 = \bra{R} R^k_{\vec{n}}(\phi) \ket{\psi_0}, \label{exact_cond1}
\end{equation}
which is equivalent to Eq.~\eqref{first_eq}, when $\vec{n},\phi$ are substituted by the corresponding rotation parameters of $F'_{xr,\alpha}$ or $F'_{xr,\beta}$ shown in Lemma~\ref{lem:decomp_fxr}.
By Eq.\eqref{exponential} in Lemma~\ref{lem:Bloch}.(3), the exponentiation $[R_{\vec{n}}(\phi)]^k$ is equal to 
\begin{equation}
R_{\vec{n}}(k\phi)
=\begin{bmatrix}
c_{k\phi} -is_{k\phi}n_z & -is_{k\phi}n_x -s_{k\phi}n_y \\
-is_{k\phi}n_x +s_{k\phi}n_y & c_{k\phi} +is_{k\phi}n_z
\end{bmatrix}.
\end{equation}
Substituting it into Eq.~\eqref{exact_cond1}, and letting its real and imaginary parts both to be zero, we obtain the following under-determined system of equations that parameters $(\beta_1,\beta_2,k)$ or $(\alpha_1,\alpha_2,k)$ need to satisfy:
\begin{align}
0 &= \sqrt{1-\lambda} \cos\frac{k\phi}{2} - \sqrt{\lambda} \sin\frac{k\phi}{2}\ n_y, \label{real1} \tag{real}\\
0 &= \sqrt{1-\lambda}\ n_z + \sqrt{\lambda}\ n_x. \label{im1} \tag{imag}
\end{align}

To be more precise, by substituting Eq.~\eqref{eq:fxr_alpha_ny} of $s_\phi n_y$ into Eq.~\eqref{real1}; and by substituting Eqs.~\eqref{eq:fxr_alpha_nx}, \eqref{eq:fxr_alpha_nz} of $s_\phi n_x, s_\phi n_z$ into Eq.~\eqref{im1}, and after some rearrangement (dividing both sides by $c_{\beta_1} c_{\beta_2}$, and then regrouping terms by $t_{\beta_1} := s_{\beta_1}/c_{\beta_1}$ and $t_{\beta_2} := s_{\beta_2}/c_{\beta_2}$, and finally multiplying back $c_{\beta_1} c_{\beta_2}$), we obtain the explicit equations that $(\beta_1,\beta_2,k)$ need to satisfy in Theorem~\ref{thm:main} as follows:
\begin{eqnarray}
0 &=& \sin\frac{\phi}{2} \cos\frac{k\phi}{2} - 2\lambda\sin\frac{k\phi}{2} [-\sin\alpha\cos\frac{\beta_1}{2} \sin\frac{\beta_2}{2} + (1-2\lambda) (1-\cos\alpha) \sin\frac{\beta_1}{2} \sin\frac{\beta_2}{2}], \label{eq:intro_fxr_alpha_real}\\
0 &=& -\sin\frac{\beta_1}{2} \sin\frac{\beta_2}{2} \sin\alpha (1-2\lambda) +\sin\frac{\beta_1}{2} \cos\frac{\beta_2}{2} [(1-2\lambda)\cos\alpha +2\lambda]
\notag\\&&
+\cos\frac{\beta_1}{2} \sin\frac{\beta_2}{2} \cos\alpha + \cos\frac{\beta_1}{2} \cos\frac{\beta_2}{2}\sin\alpha, \label{eq:intro_fxr_alpha_im}
\end{eqnarray}
where angle $\phi$ satisfies Eq.~\eqref{eq:fxr_alpha_phi}.

Similarly, the explicit equations that $(\alpha_1,\alpha_2,k)$ need to satisfy in Theorem~\ref{thm:main2} are the following:
\begin{eqnarray}
0 &=& \sin\frac{\phi}{2} \cos\frac{k\phi}{2} - 2\lambda\sin\frac{k\phi}{2} [-\sin\frac{\alpha_1}{2} \cos\frac{\alpha_2}{2} \sin\beta +(1-2\lambda) \sin\frac{\alpha_1}{2} \sin\frac{\alpha_2}{2} (1-\cos\beta)], \label{eq:intro_fxr_beta_real}\\
0 &=& -\sin\frac{\alpha_1}{2} \sin\frac{\alpha_2}{2} (1-2\lambda)\sin\beta +\sin\frac{\alpha_1}{2} \cos\frac{\alpha_2}{2} \cos\beta
\notag\\&&
-\cos\frac{\alpha_1}{2} \sin\frac{\alpha_2}{2} (2\lambda+ (1-2\lambda)\cos\beta) +\cos\frac{\alpha_1}{2} \cos\frac{\alpha_2}{2} \sin\beta, \label{eq:intro_fxr_beta_im}
\end{eqnarray}
where angle $\phi$ satisfies Eq.~\eqref{eq:fxr_beta_phi}.

\section{Existence of parameters when $k>k_{\text{lower}}$} \label{sec:iterations k}
In this section we prove the sufficient condition for our FXR method to succeed: whenever the number of iterations $k$ is fixed to be greater than $k_{\mathrm{lower}}$ [see Eq.~\eqref{k_lower_intro} or \eqref{k_lower2_intro}], there exits a solution $(\beta_1,\beta_2)$ or $(\alpha_1,\alpha_2)$ to the system of equations \eqref{real1}, \eqref{im1}. More explicitly, Eqs.~\eqref{eq:intro_fxr_alpha_real}, \eqref{eq:intro_fxr_alpha_im} for $(\beta_1,\beta_2)$; and Eqs.~\eqref{eq:intro_fxr_beta_real}, \eqref{eq:intro_fxr_beta_im} for $(\alpha_1,\alpha_2)$. Their value can be determined numerically by using, for example, MATLAB.

We first prove for the case when $\alpha$ is fixed (Theorem~\ref{thm:main}) from Section~\ref{subsec:continuous} to \ref{subsec:lower_k}. The proof for the case when $\beta$ is fixed (Theorem~\ref{thm:main2}) is almost the same, but we point out some differences in Section~\ref{subsec:differ}.

The outline of our proof for the existence of solution $(\beta_1,\beta_2)$ to Eqs.~\eqref{real1}, \eqref{im1} whenever $k$ is greater than $k_{\text{lower}}$ [see Eq.~\eqref{k_lower_intro}] is as follows, one may find it helpful before delving into the details.
\begin{enumerate}
    \item In Section~\ref{subsec:continuous} we first deal with Eq.~\eqref{im1}, which does not contain $k$. Thus solving it will give us a family of continuous functions $\beta_2=f(\beta_1)+2l\pi$ for $l\in\mathbb{Z}$, which allows us to regard the RHS of Eq.~\eqref{real1} as a univariate function of $\beta_1$.
    \item Let $R_{\vec{n}}(\phi) := F'_{xr,\alpha}$, we prove in Lemma~\ref{thm:phi zero} that there exists a pair of special parameters $(\beta_1',\beta_2')$ satisfying Eq.~\eqref{im1} (so we can denote the branch of continuous functions $f+2l\pi$ that pass through $(\beta_1',\beta_2')$ by $\hat{f}$ ) such that $\phi=0$. And in Lemma~\ref{lem:another_pair} we find another pair of parameters $(\beta_1'',\beta_2'')$ on $\hat{f}$ such that $\phi/2 \in \{\phi_0, \pi-\phi_0\}$, where $\phi_0$ is shown by Eq.~\eqref{phi_0}.
    \item Then in Section~\ref{subsec:lower_k}, by substituting $\beta_2=\hat{f}(\beta_1)$ into Eq.~\eqref{eq:fxr_alpha_phi} about $c_\phi$, we get another continuous function regarding $\phi$: $\phi/2=g(\beta_1)$, and from step 2 we know that $g(\beta_1')=0$ and $g(\beta_1'')\geq \phi_0$. Thus by the continuity of $g$, when $k$ is big enough that $\pi/k \leq \phi_0$ [this inequality gives us the lower bound on $k$ shown in Eq.~\eqref{k_lower_intro}], there exists $\beta_1'''$ such that $\phi/2=g(\beta_1''')=\pi/k$.
    \item Finally, by substituting the above two continuous functions $\beta_2=\hat f(\beta_1)$ and $\frac{\phi}{2}=g(\beta_1)$ into the RHS of Eq.~\eqref{real1}, we obtain a continuous function $F(\beta_1)$. Since $\frac{\phi}{2}=g(\beta_1')=0$ and $\frac{\phi}{2}=g(\beta_1''')=\pi/k$ from step~3, we know that $F(\beta_1')>0$ and $F(\beta_1''')<0$. Therefore by the intermediate value theorem of continuous function on closed interval, we conclude that there exists $\hat{\beta}_1$ such that $F(\hat{\beta}_1)=0$, which completes the proof.
\end{enumerate}

\subsection{ Eq.~\eqref{im1} leads to continuous function $\beta_2=f(\beta_1)$ }\label{subsec:continuous}
Consider Eq.~\eqref{im1}, or equivalently Eq.~\eqref{eq:intro_fxr_alpha_im},
we find that $t_2:=\tan\frac{\beta_2}{2}$ is a function of $t_1:=\tan\frac{\beta_1}{2}$ as follows.
\begin{equation}
    t_2 = -\frac{t_1\big[ (1-2\lambda) c_{2\alpha} + 2\lambda \big] + s_{2\alpha}}{c_{2\alpha} - t_1 s_{2\alpha} (1-2\lambda)} \tag{2f1} \label{2f1}
\end{equation}

Note that when $\beta_1$ satisfies $t_1=\frac{c_{2\alpha}}{s_{2\alpha}(1-2\lambda)}$, the denominator of the RHS of Eq.~\eqref{2f1} will become zero. We call this $\beta_1$ the potential `discontinuity point'.
If we define $\beta_2=f(\beta_1)$ by taking `$2\arctan$' of the RHS of Eq.~\eqref{2f1}, then because the range of $2\arctan$ is $(-\pi,\pi)$, $f$ will jump from $-\pi$ to $\pi$, or vice-versa, at this `discontinuity point'.
This phenomenon is shown in Fig. \ref{fig:continum_f} by the discontinuous black curve between the two horizontal red lines, where the red lines represent the range of $2\arctan$.
However, one can always translate the continuous segment on one side of the `discontinuity point' up or down by $2\pi$, and concatenate it with the other side, to form a new continuous function, which we still denote by $\beta_2=f(\beta_1)$.
Thus, we say that function $\beta_2=f(\beta_1)$ is {\it inherently continuous}. A more rigorous proof considering different cases of potential `discontinuity point' can be found in Lemma~\ref{lem:continuum_f}.

\begin{figure}[ht]
\centering
\includegraphics[width=0.5\textwidth]{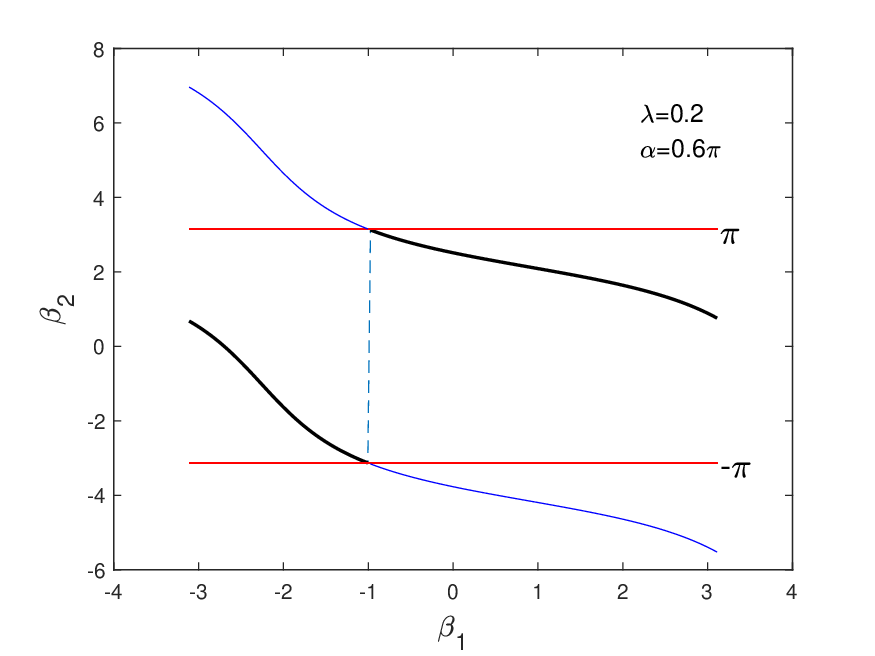}
\caption{\label{fig:continum_f} Image of function $\beta_2=f(\beta_1)$ when $\lambda=0.2, \alpha=0.6\pi$. Black curve between the two horizontal red line represents $f$ obtained by taking $\arctan$ of Eq.~\eqref{2f1}, blue curve outside the two horizontal line represents translation of continuous segment of $f$ to form new continuous $f$. }
\end{figure}

Because the period of $\tan x$ is $\pi$, we have actually obtained a family of continuous functions $\beta_2=f(\beta_1)+2l\pi,\ \beta_1\in[-\pi,\pi]$ which satisfy Eq.~\eqref{2f1}, where $l\in\mathbb{Z}$.
In other words, the solutions $(\beta_1,\beta_2)$ to Eq.~\eqref{im1} constitutes a family of parallel continuous function (curve) $\beta_2=f(\beta_1)+2l\pi,\, \forall l\in\mathbb{Z}$ in the two-dimensional plane of $(\beta_1,\beta_2)$.

\subsection{ A pair of special parameters $(\beta_1',\beta_2')$ }\label{subsec:special beta 1}
In this subsection, we present our key observation (Lemma~\ref{thm:phi zero}) that there exists a pair of parameters $(\beta_1',\beta_2')$ satisfying Eq.~\eqref{im1} such that the rotation angle $\phi$ of $R_{\vec{n}}(\phi)$ is zero! This is done by reasoning with 3D geometry rather than solving Eqs.\eqref{eq:intro_fxr_alpha_im}, \eqref{eq:fxr_alpha_phi} directly (which is likely to be intractable).

\begin{lemma}\label{thm:phi zero}
There exists a pair of parameters $(\beta_1',\beta_2')$ satisfying Eq.~\eqref{im1} such that $R_{\vec{n}}(\phi) := F'_{xr,\alpha}$ is the identity transformation $I$, which implies $\phi=0$.
\end{lemma}

\begin{proof}
The proof is divided into four steps:

1. We first show that there exists a specific angle $\beta_1'$ such that $R_{\vec{n}}(\phi)$ keeps $\ket{\psi_0}$ still as follows.

From Lemma~\ref{lem:S_o decomp} we know that $G'(\alpha,\beta) = R_{\psi_0}(\alpha) R_{R}(\alpha)$ represents: a rotation about $\ket{R}$ by angle $\alpha$, followed by a rotation about $\ket{\psi_0}$ by angle $\beta$ on the Bloch sphere, which are illustrated by the two intersecting small blue circle shown in Fig. \ref{fig:bloch_sphere}.(b).  

\begin{figure*}[ht]
\centering
\subfloat{\includegraphics[width=0.43\textwidth]{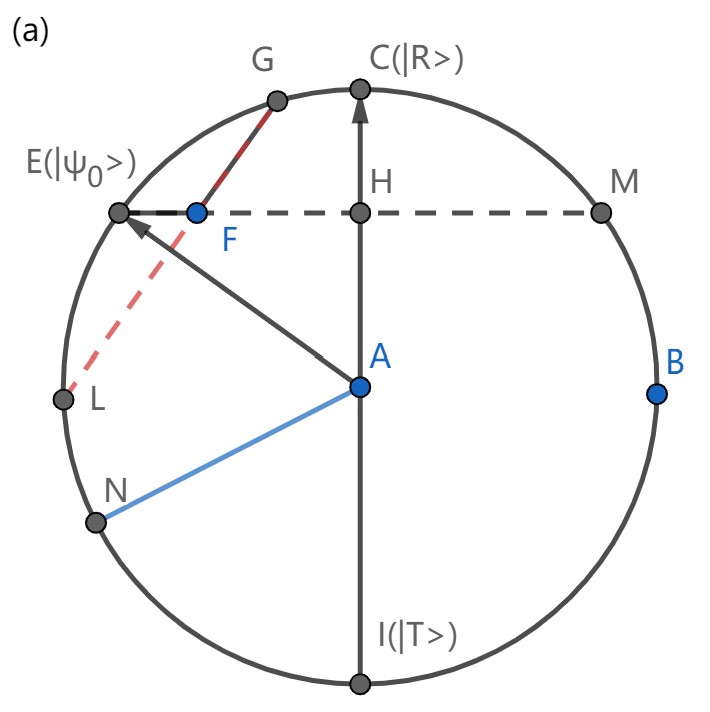}}
\subfloat{\includegraphics[width=0.5\textwidth]{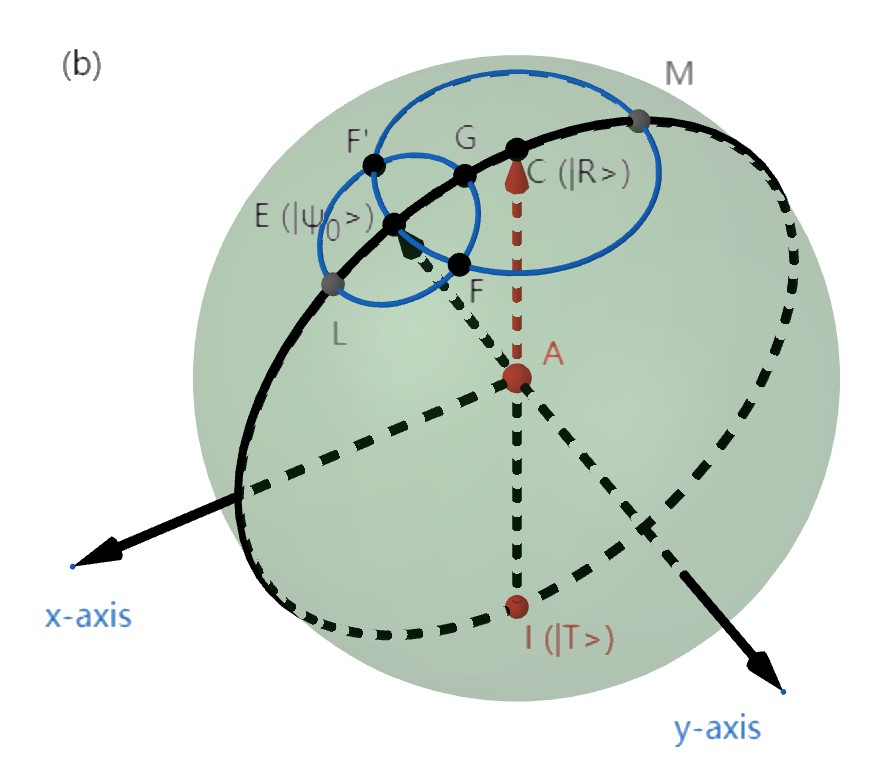}}
\caption{\label{fig:bloch_sphere} Bloch sphere interpretation on the effect of $R_{\psi_0}(\beta)$ and $R_{R}(\alpha)$ on the initial state $E (\ket{\psi_0})$. (a) Perspective of (b) from the positive $y$-axis. $\mathop{AN}\limits^{\longrightarrow}$ bisects $\angle EAI$. (b) Bigger blue circle $EFMF'E$ represents the trajectory of $E$ under rotation $R_{R}(\alpha)$. Smaller blue circle $FGF'LF$ represents the trajectory of $F$ under rotation $R_{\psi_0}(\beta)$.  }
\end{figure*}

Changing perspective to look at the $xz$-plane from the positive $y$-axis, one gets the view shown by Fig. \ref{fig:bloch_sphere}.(a).
The horizontal dashed black line $\overline{EM}$ represents the circular trajectory of $E (\ket{\psi_0})$ under rotation $R_{R}(\alpha)$. The oblique dashed red line $\overline{GL}$ (which is perpendicular to $\mathop{AE}\limits^{\longrightarrow}$) represents the circular trajectory of $F$ under rotation $R_{\psi_0}(\beta_1)$.
We now consider the trajectory of $E$ under operation $R_{\vec{n}}(\phi)=G'(\alpha,\beta_2)G'(\alpha,\beta_1)$. 
Firstly, suppose $R_{R}(\alpha)$ rotates $E$ to $F$ on $\overline{EM}$. 
Then, applying $R_{\psi_0}(\beta_1)$ rotates $F$ to any point on the oblique red dashed line $\overline{GL}$. So there exists an angle $\beta_1=\beta_1'$ such that $R_{\psi_0}(\beta_1')$ rotates $F$ to $F'$, which is the the other intersecting point of `circle' $\overline{EM}$ and `circle' $\overline{GL}$ (shown in Fig. \ref{fig:bloch_sphere}.(b)). In the view of Fig. \ref{fig:bloch_sphere}.(a), $F'$ coincides with $F$.
Finally, $R_{R}(\alpha)$ rotates $F'$ back to $E$, which is on the axis of rotation $R_{\psi_0}(\beta_2')$, so applying $R_{\psi_0}(\beta_2')$ keeps $E$ still. In the view of Fig. \ref{fig:bloch_sphere}.(b), the overall effect of $R_{\vec{n}}(\phi)$ on $E$ is the trajectory $\overset{\frown}{EF}-\overset{\frown}{FGF'}-\overset{\frown}{F'E}-E$.

2. We now show that Eq.~\eqref{im1} holds if and only if the rotation axis $\vec{n}$ of $R_{\vec{n}}(\phi)$ satisfies $\braket{T|\vec{n}}=\braket{\psi_0|\vec{n}} $, which is equivalent to $\vec{n}$ being on the perpendicular bisectionplane of $\overline{EI}$.

Note that on the Bloch sphere, $\ket{T}=[0,0,-1] $ and $\ket{\psi_0} = [\sin\theta,0,\cos\theta]=[ 2\sqrt{\lambda(1-\lambda)},0,(1-2\lambda) ]$, so $\braket{T|\vec{n}}=\braket{\psi_0|\vec{n}}$ is equivalent to 
$-n_z = 2\sqrt{\lambda(1-\lambda)} n_x + (1-2\lambda) n_z$, which is equivalent to $0 = \sqrt{1-\lambda}\ n_z + \sqrt{\lambda}\ n_x $, and it's exactly Eq.~\eqref{im1}.

Next we would like to show that the rotation axis $\vec{n}$ of $R_{\vec{n}}(\phi)$ is on the perpendicular bisectionplane of $\overline{EI}$, if and only if $ \braket{T|\vec{n}}=\braket{\psi_0|\vec{n}} $. 
Suppose $\vec{n}$ is represented by $\mathop{AN}\limits^{\longrightarrow}$. Then the former condition is equivalent to $|NE|=|NI|$, because the bisectionplane is all the points in the $\mathbb{R}^3$ space that are equidistant from $E$ and $I$.
Suppose $|NE|=|NI|$, then we have $\triangle ANE \cong \triangle ANI $ (note that these two triangle might not be coplanar, because $\vec{n}$ is not necessarily in the xz-plane), because all their sides are equal to $1$. Thus the projections of $AE,AI$ to $AN$ are of the same length, which implies $\braket{T|\vec{n}}=\braket{\psi_0|\vec{n}} $, finishing the `only if' part.
Conversely, if the projections are of the same length, then it can be seen that the two right triangles formed by projection are congruent (note that $|AE|=|AI|=1$), so $\angle EAN=\angle IAN$. Thus we have $\triangle EAN \cong \triangle IAN $. Therefore, $|NE|=|NI|$ which completes the `if' part.

3. Take any branch of the continuous function family $\beta_2=f(\beta_1)+2l\pi$, and let $\beta_2'=f(\beta_1')$. Then the pair of parameters $\beta_1',\beta_2'$ satisfy Eq.~\eqref{im1}, and thus by step 2 we know the rotation axis $\vec{n}\, (\mathop{AN}\limits^{\longrightarrow})$ of $R_{\vec{n}}(\phi)$ is not collinear with $\ket{\psi_0} (E)$. Because $R_{\vec{n}}(\phi)$ keeps $\ket{\psi_0}$ still by step 1, and that any nontrivial three-dimensional rotation can only keep its axis still, we conclude that the parameters $\beta_1',\beta_2'$ makes $R_{\vec{n}}(\phi)$ become the trivial transformation $I$.

4. From  Eq.~\eqref{eq:fxr_alpha_phi}, one would notice that when one of the $\beta_i$ is added by $2\pi$, the sign of $\cos\phi$ would change. So there is a minor problem in step 3, that is, $R_{\vec{n}}(\phi)$ might be $\pm I$. But adding $2\pi$ to $\beta_i$ does not change the geometry effect of rotation in the Bloch sphere. Thus the analysis in step 3 still works, except that $R_{\vec{n}}(\phi)$ might be $-I$. But we can always let $\beta_2'\leftarrow \beta_2'+2k\pi$, where $k\in\{0,1\}$ such that $\cos\phi=+1$. Then this new pair of $(\beta_1',\beta_2')$ makes $R_{\vec{n}}(\phi) = I$ and the angle $\phi=0$.
\end{proof}

Because $(\beta_1',\beta_2')$ in Lemma~\ref{thm:phi zero} satisfy Eq.~\eqref{im1}, we can always find a branch $\hat{f}$ of $\beta_2=f(\beta_1)+2l\pi$ such that $\beta_2'=\hat f(\beta_1')$. This continuous function $\hat{f}$ will be frequently used in the rest of this section.

\subsection{Another pair of special parameters $(\beta_1'',\beta_2'')$ }\label{subsec:theta2}

\begin{lemma}\label{lem:another_pair}
Let $\beta_1''= -\alpha$ and $\beta_2''=\hat f(\beta_1'')$, then the rotation angle $\phi$ of $R_{\vec{n}}(\phi)$ satisfy $\cos\frac{\phi}{2} \in\{\cos4\phi',-\cos4\phi'\}$.
\end{lemma}

\begin{proof}
Inspired by Ref. \cite{Long}, we find that when the two phase angles have opposite sign, i.e. $G'(\alpha,-\alpha) = R_{\psi_0}(-\alpha) R_{R}(\alpha)$, the rotation axis of $G'(\alpha,-\alpha)$ satisfies Eq.~\eqref{im1}, and the rotation angle $\phi_d$ is $4\phi'$ with
\begin{equation}
\sin\phi'=\sqrt{\lambda}\sin\frac{\alpha}{2}. \label{phi'}
\end{equation}
In fact, from Eq.~\eqref{eq:G_alpha_angle} we have $c_{\phi_d} = c^2_\alpha + s^2_\alpha(1-2\lambda) = 1-2\lambda s^2_\alpha = \cos 2\phi'$. From Eq.~\eqref{eq:G_alpha_axis}, we know $\vec{n}_d$ is collinear with $\left[ -\sin\alpha \sin\theta,\ (1-\cos\alpha) \sin\theta,\ \sin\alpha(1-\cos\theta) \right]$, which is collinear with 
$$[-\sqrt{1-\lambda} \sin\alpha,\ \sqrt{1-\lambda}\, (1-\cos\alpha),\ \sqrt{\lambda}\sin\alpha ].$$
So $\sqrt{1-\lambda}\, n_{dz} + \sqrt{\lambda}\, n_{dx}=0$, which is Eq.~\eqref{im1}.

Let $\beta_1''= -\alpha$ and $\beta_2''=\hat f(\beta_1'')$. Then 
the rotation axis of $R_{\vec{n}}(\phi)=G'(\alpha,\beta_2'') G'(\alpha,\beta_1'')$ satisfy Eq.~\eqref{im1} and the rotation angle is $\phi=8\phi'$, because the constituent two rotations have the same axis satisfying Eq.~\eqref{im1}, and both have rotation angle $4\phi'$. Similar to step 4 in the proof of Lemma~\ref{thm:phi zero}, from  Eq.~\eqref{eq:fxr_alpha_phi} of $\cos\frac{\phi}{2}$ we know that when one of the $\beta_i$ is added by $2\pi$, the sign of $c_\phi$ would change. Thus we can only guarantee that $(\beta_1'',\beta_2'')$ causes the rotation angle $\phi$ of $R_{\vec{n}}(\phi)$ to satisfy $c_\phi \in\{\cos4\phi',-\cos4\phi'\}$.
\end{proof}

\subsection{Lower bound of  $k$}\label{subsec:lower_k}
From the above three subsections, we know that there exists a continuous function $\beta_2=\hat f(\beta_1)$ defined on the closed interval $[-\pi,\pi]$ such that:
\begin{enumerate}
    \item the pair of variable parameters $(\beta_1,\beta_2)$ determined by $\hat{f}$ satisfy Eq.~\eqref{im1};
    \item there exists $\beta_1'$ such that the rotation angle $\phi=0$;
    \item there exists $\beta_1''$ such that $\cos\frac{\phi}{2} \in\{\cos4\phi',-\cos4\phi'\}$.
\end{enumerate}

Substitute $\beta_2=\hat f(\beta_1)$ into Eq.~\eqref{eq:fxr_alpha_phi} of $\cos\frac{\phi}{2}$, and taking $\arccos$ on both sides, we obtain another continuous function $\frac{\phi}{2}=g(\beta_1)$ with range $[0,\pi]$. Let
\begin{equation}
    \phi_0:=\arccos|\cos 4\phi'| \in [0,\frac{\pi}{2}]. \label{phi_0}
\end{equation}
Then $\arccos(-|\cos 4\phi'|) = \pi-\phi_0 \geq \phi_0$, and thus $g(\beta_1'')\geq \phi_0$. Note that $g(\beta_1')=\arccos(1)=0$.
Therefore, from the continuity of function $g$, the range of $g(\beta_1)$ must contain the interval $[0,\phi_0]$ by the intermediate value theorem. Note that it may not contain $\pi-\phi_0$, as Fig. \ref{fig:phi_g_theta} shows.

\begin{figure}[ht]
\centering
\includegraphics[width=0.5\textwidth]{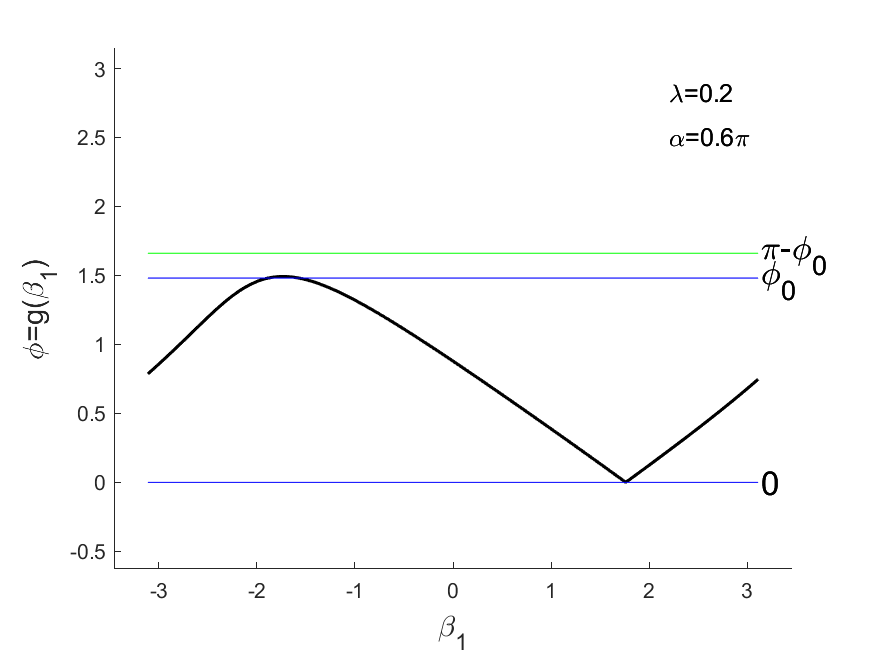}
\caption{\label{fig:phi_g_theta} Image of $\frac{\phi}{2}=g(\beta_1),\ \beta_1\in[-\pi,\pi]$ when $\lambda=0.2, \alpha=0.6\pi$. In this example, the range of $g$ contains interval $[0,\phi_0]$ but does not contain $\pi-\phi_0$. }
\end{figure}

Because the range of $g(\beta_1)$ contains $[0,\phi_0]$, when the number of iterations $k$ is big enough such that
\begin{equation}
    \pi/k \leq \phi_0 \label{k_condition},
\end{equation}
there exits 
$\beta_1'''$ such that $\pi/k=g(\beta_1''')$. 
Recall that the equation obtained from letting the real part of Eq.~\eqref{exact_cond1} to be zero is 
\begin{equation}
    0 = \sqrt{1-\lambda} \cos{(k\frac{\phi}{2})} - \sqrt{\lambda} \sin{(k\frac{\phi}{2})}\ n_y, \tag{real 1} \label{real 1.1}
\end{equation}
where $n_y$ satisfies 
$$
s_\phi n_y = 2\sqrt{\lambda(1-\lambda)} \Big( -\sin\alpha \cos\frac{\beta_1}{2} \sin\frac{\beta_2}{2} + 2(1-2\lambda) \sin^2\frac{\alpha}{2} \sin\frac{\beta_1}{2} \sin\frac{\beta_2}{2} \Big),
$$
from Eq.~\eqref{eq:fxr_alpha_ny}.
Substituting the two continuous function $\beta_2=\hat f(\beta_1)$ and $\frac{\phi}{2}=g(\beta_1)$ into the RHS of Eq.~\eqref{real 1.1}, we obtain a continuous function $F(\beta_1)$ defined on the closed interval $[-\pi,\pi]$.
Note that $s_{k\phi}/s_{\phi}$ has limit $1$ when $\phi\to 0$, thus $F$ is actually continuous at $\phi=0$.

To sum up, we have
\begin{enumerate}
    \item $F(\beta_1')=\sqrt{1-\lambda}>0$, since $\frac{\phi}{2}=g(\beta_1')=0$;
    \item $F(\beta_1''')=-\sqrt{1-\lambda}<0$, since $\frac{\phi}{2}=g(\beta_1''')=\pi/k$.
\end{enumerate}
Therefore, by the intermediate value theorem of continuous function on a closed interval, there exists $\hat{\theta}_1$ such that $F(\hat{\beta}_1)=0$. Let $\hat\beta_2=\hat f(\hat\beta_1)$. Then $(\hat\beta_1,\hat\beta_2)$ satisfies the system of equations $\eqref{im1}, \eqref{real1}$.

It is easy to show that $\arccos |\cos x| = \big|x\  \mathrm{mod} [-\frac{\pi}{2}, \frac{\pi}{2}] \big| $. Recall from Eqs. \eqref{phi_0}, \eqref{phi'} that $\phi_0 = \arccos|\cos 4\phi'|$ where $\phi'$ satisfies $\sin\phi'=\sqrt{\lambda}\sin\frac{\alpha}{2}$. Therefore by Eq.~\eqref{k_condition}, whenever  $k$, the number  of iterations of  $R_{\vec{n}}(\phi)$, is greater than the following $k_{\mathrm{lower}}$, there exits a solution $(\beta_1,\beta_2)$ to the system of equations \eqref{real1}, \eqref{im1}. This completes the proof of Theorem~\ref{thm:main}.
\begin{equation}\label{eq:k_lower_alpha}
    k_{\mathrm{lower}} = \frac{\pi}{ \Big| 4\arcsin(\sqrt{\lambda}\sin\frac{\alpha}{2}) \mod [-\frac{\pi}{2},\frac{\pi}{2}] \Big| }
\end{equation}


\subsection{Differences in proof when $\beta$ is fixed}\label{subsec:differ}
We now consider the case when $\beta$ in $G(\alpha,\beta)$ is fixed. It is almost the same as in the previous three subsections to prove the existence of solution $(\alpha_1,\alpha_2)$ to Eqs.~\eqref{real1}, \eqref{im1}, or to be more precise, Eqs.~\eqref{eq:intro_fxr_beta_real}, \eqref{eq:intro_fxr_beta_im}, when the number of iterations $k$ of $F_{xr,\beta}=G(\alpha_2,\beta) G(\alpha_1,\beta)$ is greater than $k_{\text{lower}}$ as shown in Eq.~\eqref{k_lower2_intro}. Because we are treating the iteration $F_{xr,*}$ as a whole.

However, there are at least three differences in the proof worth mentioning: 
\begin{enumerate}
    \item In parallel to Section~\ref{subsec:continuous}, we need to obtain the continuous function $\alpha_1=f(\alpha_2)$ from Eq.~\eqref{im1} rather than $\alpha_2=f(\alpha_1)$. This is because in the next step corresponding to Section~\ref{subsec:special beta 1}, we will find a specific angle $\alpha_2'$ such that $R_{\vec{n}}(\phi)$ keeps $\ket{R}$ still. So $\alpha_2$ needs to be the variable parameter of $(\alpha_1,\alpha_2)$ in our proof.
    \item Rather than showing that there exits a special angle $\beta_1'$ such that $R_{\vec{n}}(\phi)$ keeps $\ket{\psi_0}$ still, we will instead prove that there exists an angle $\alpha_2'$ such that $R_{\vec{n}}(\phi)$ keeps $\ket{R}$ still, this follows easily from Fig.~\ref{fig:2d_rotate}.
    \item Parallel to Section~\ref{subsec:theta2}, we will instead let $\alpha_2''=-\beta$, and Eq.~\eqref{phi'} now becomes
    \begin{equation}
        \sin\phi'=\sqrt{\lambda}\sin\frac{\beta}{2}.
    \end{equation}
    Thus, the lower bound $k_{\text{lower}}$ with $\beta$ being fixed [see Eq.~\eqref{k_lower2_intro}] is obtained by changing $\alpha$ to $\beta$ in Eq.~\eqref{eq:k_lower_alpha}.
\end{enumerate}

\begin{figure}[ht]
\centering

\definecolor{xdxdff}{rgb}{0.49019607843137253,0.49019607843137253,1.}
\definecolor{uuuuuu}{rgb}{0.26666666666666666,0.26666666666666666,0.26666666666666666}
\begin{tikzpicture}[line cap=round,line join=round,x=1.0cm,y=1.0cm,scale=4]
\draw [line width=0.4pt] (0.,0.) circle (1.cm);
\draw [line width=0.4pt] (0.,0.)-- (0.,1.2);
\draw [line width=0.4pt] (0.,0.)-- (-0.5770512675001516,1.0521462990841475);
\draw [line width=0.4pt,dotted] (-0.7091458651201455,0.705061800116132)-- (0.7091458651201455,0.705061800116132);
\draw [line width=0.4pt,dashed] (-0.8432532715029178,0.5375164370506507)-- (0.,1.);
\begin{scriptsize}
\draw [fill=uuuuuu] (0.,0.) circle (0.5pt);
\draw[color=uuuuuu] (0.08719610163946265,0.015282339062338599) node {$A$};
\draw [fill=xdxdff] (0.,1.2) circle (0.5pt);
\draw[color=xdxdff] (0.05960366111800694,1.2508123631547141) node {$\ket{R}$};

\draw [line width=0.4pt] (0.,0.)-- (0.,-1.1);
\draw [fill=xdxdff] (0,-1.1) circle (0.5pt);
\draw[color=xdxdff] (0.1,-1.1) node {$\ket{T}$};

\draw [fill=xdxdff] (-0.5770512675001516,1.0521462990841475) circle (0.5pt);
\draw[color=xdxdff] (-0.5106401096587445,1.1220474474924813) node {$\ket{\psi_0}$};

\draw [fill=uuuuuu] (0.,1.) circle (0.5pt);
\draw[color=uuuuuu] (-0.0749327947410379,1.0791258089384037) node {$E$};
\draw [fill=uuuuuu] (-0.8432532715029178,0.5375164370506507) circle (0.5pt);
\draw[color=uuuuuu] (-0.7743012079748769,0.48128870050660927) node {$E$};
\draw [fill=uuuuuu] (-0.5377652783099913,0.705061800116132) circle (0.5pt);
\draw[color=uuuuuu] (-0.48611349586189506,0.5732636402653468) node {$F (F')$};
\draw [fill=uuuuuu] (-0.7091458651201455,0.705061800116132) circle (0.5pt);
\draw[color=uuuuuu] (-0.808025301945545,0.72) node {$G$};
\draw [fill=uuuuuu] (0.7091458651201455,0.705061800116132) circle (0.5pt);
\draw[color=uuuuuu] (0.6789006594884576,0.6253827727952982) node {$L$};

\draw [fill=uuuuuu] (-0.48,0.88) circle (0.5pt);
\draw[color=uuuuuu] (-0.465,0.95) node {$C$};

\draw[->,color=red] (0.02,0.98) .. controls (0.12,0.92) and (0.12, 1.08) .. (0.02,1.02);
\draw[color=red] (0.12,1.05) node {$\alpha_1$};

\draw[->,color=red] (-0.05,1) --  (-0.5377652783099913,0.74);
\draw[color=red] (-0.29,0.91) node {$\beta$};

\draw[->,color=red] (-0.55,0.74) .. controls (-0.82,0.8) and (-0.82,0.6) .. (-0.55,0.66);
\draw[color=red] (-0.68,0.8) node {$\alpha_2$};

\draw[<-,color=red] (-0.05,0.94) --  (-0.5377652783099913,0.68);
\draw[color=red] (-0.29,0.75) node {$\beta$};

\end{scriptsize}
\end{tikzpicture}

\caption{\label{fig:2d_rotate} Bloch sphere interpretation of the effect of $F_{xr,\beta}=S_r(\beta) S_o(\alpha_2) S_r(\beta) S_o(\alpha_1)$ on state vector $\ket{R}=[0,0,1]$. This is a view of the Bloch sphere along its $y$-axis.
$S_o(\alpha)$ represents a rotation about $\ket{R}$ (i.e. $\overline{AE}$) by $\alpha$, and $S_r(\beta)$ represents a rotation about $\ket{\psi_0}$ (i.e. $\overline{AC}$) by $\beta$.
The four red arrow lines represents trajectory $E-\overset{\frown}{EF'}-\overset{\frown}{F'GF}-\overset{\frown}{FE}$ of $\ket{R}$ under operator $F_{xr,\beta}$, and shows that there exists an angle $\alpha_2$ such that $F_{xr,\beta}$ keeps $\ket{R}$ still.}
\end{figure}
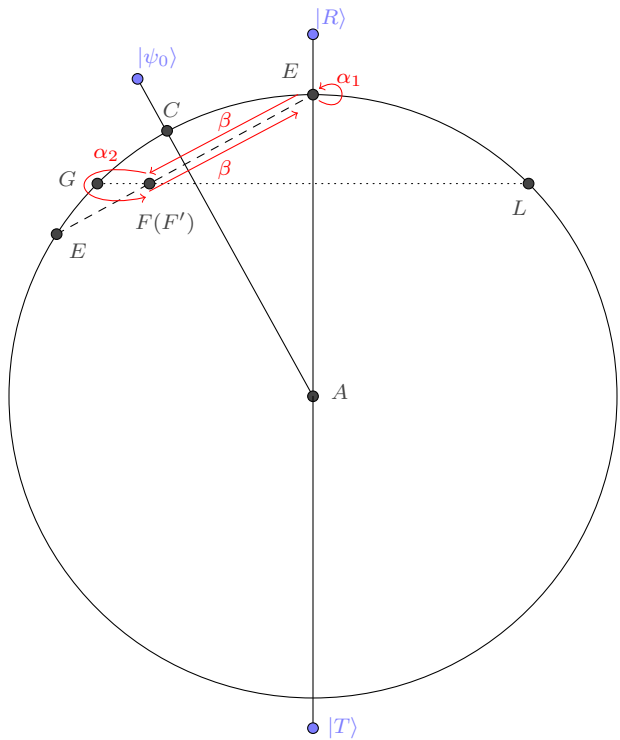

\section{Conclusion}\label{sec:conclusion}
In this paper we have proposed a  search framework with adjustable parameters, in which given the phase oracle $S_o(\alpha)$ with an arbitrary angle $\alpha$ or the phase rotation $S_r(\beta)$  with an arbitrary angle $\beta$, we can always construct a quantum exact search algorithm without sacrificing the quadratic speedup advantage. In technique, 
we have  proposed the fixed-axis-rotation (FXR) method which computes the state of a two-dimensional system in a concise way.
Two applications of the proposed search framework have been developed. 
Possible future research may include extending this framework to fixed-point or robust quantum search \cite{fixed_point,robust}, as well as looking for more applications.

\bibliographystyle{unsrt}
\bibliography{sample}


\begin{appendices}
\section{Proof of Lemma~\ref{lem:S_o decomp}}\label{proof:S_o}
{\it Proof of Lemma~\ref{lem:S_o decomp}:} Suppose $S_o(\alpha)=e^{i\varphi} R_{\vec{n}}(\phi)$ [see their definition in Eqs.~\eqref{S_o},\eqref{R n matrix}], tracing (i.e. summing up diagonal elements of) the matrix on both sides, we have $1+e^{i\alpha}=2\cos\phi\, e^{i\varphi}$. By parallelogram rule of adding two complex value, the left hand side (LHS) has phase angle $\alpha$ (angle between the positive $x$-axis and a ray from the origin to the point $(x,y)$ representing complex value $z=x+iy$ in the $xy$-plane), so we can take $\varphi=\alpha/2$. Note that $S_o(\alpha)$ is diagonal, so it's obvious that $e^{-i\alpha/2} S_o(\alpha)=\cos\frac{\alpha}{2}\ I -i\sin\frac{\alpha}{2}\  Z$, which by Lemma~\ref{lem:Bloch}.(2) represents in the Bloch sphere a rotation about $\ket{R}=[0,0,1]^T$ by an angle of $\alpha$.

Similarly, tracing both sides of $S_r(\beta)=e^{i\varphi} R_{\vec{n}}(\phi)$ [defined in Eq.~\eqref{s_r}], we have $1+e^{-i\beta}=2\cos\phi\, e^{i\varphi}$, so $\varphi=-\beta/2$. By comparing the two sides of equation $R_{\vec{n}}(\phi) = e^{i\beta/2} S_r(\beta)$, we obtain the rotation angle $\phi$ as follows:
\begin{align*}
\cos\phi &= e^{i\beta/2} \frac{1}{2} (s_{11}+s_{22}) \\
&= e^{i\beta/2} \frac{1+e^{-i\beta}}{2} \\
&= \cos\frac{\beta}{2},
\end{align*}
where $s_{ij}$ is the $(i,j)$-element of $S_r(\beta)$ for $i,j\in\{1,2\}$. Similarly for the $z$-coordinate $n_z$ of the rotation axis $\vec{n}$, we have
\begin{align*}
s_\phi n_z &= e^{i\beta/2} \frac{1}{2} (s_{11}-s_{22}) i \\
&= e^{i\beta/2} \frac{-1}{2} (1-e^{-i\beta})(1-2\lambda)i \\
&= -i \sin\frac{\beta}{2} (1-2\lambda)i \\
&= \sin\frac{\beta}{2} \cos\theta.
\end{align*}
Recall that $\sqrt{\lambda}=\sin\frac{\theta}{2}$ and $\sqrt{1-\lambda}=\cos\frac{\theta}{2}$ from Eq.~\eqref{eq:lambda_trigono}, thus $\cos\theta = 1 - 2\sin^2\frac{\theta}{2} = 1 - 2\lambda$ in the last equality. As for the $x$-coordinate $n_x$ of the rotation axis $\vec{n}$: 
\begin{align*}
s_\phi n_x &= e^{i\beta/2} \frac{1}{2} (s_{12} + s_{21}) i \\
&= e^{i\beta/2} (-1)(1-e^{-i\beta}) \sqrt{\lambda(1-\lambda)}\ i\\
&= -i \sin(\frac{\beta}{2})\ 2\sqrt{\lambda(1-\lambda)}\ i \\
&= \sin\frac{\beta}{2} \sin\theta.
\end{align*}
Note that in the last equality we use $\sin\theta = 2\sin\theta \cos\theta = 2\sqrt{\lambda(1-\lambda)}$. Finally for the $y$-coordinate $n_y$ of the rotation axis $\vec{n}$: 
\[
s_\phi n_y = e^{i\beta/2} \frac{1}{2} (s_{21} - s_{12}) = 0.
\]
Therefore $e^{i\beta/2} S_r(\beta) = \cos\frac{\beta}{2}\ I -i\sin\frac{\beta}{2}(\sin\theta\ X +\cos\theta\ Z)$, which is a rotation about $\ket{\psi_0}=[\sin\theta, 0, \cos\theta]$ by an angle of $\beta$.
\hfill $\blacksquare$ \par

\section{Rotation parameters of $F'_{xr,\alpha}$ and $F'_{xr,\beta}$ }\label{app:Fxr}
We will calculate the rotation parameters $c_\phi$ and $s_\phi\vec{n}$ of $F'_{xr,\alpha} = G'(\alpha,\beta_2) G'(\alpha,\beta_1)$ and $F'_{xr,\beta} = G'(\alpha_2,\beta) G'(\alpha_1,\beta)$, filling in details of the proof outline of Lemma~\ref{lem:decomp_fxr}.

We first calculate the effect of $G'(\alpha,\beta)$ directly as follows:
\begin{align}
G'(\alpha,\beta) &= R_{\psi_0}(\beta) R_{R}(\alpha) \\
&= [c_\beta I -i\, s_\beta(s_{2\theta} X +c_{2\theta} Z)]\, [c_\alpha I -i\, s_\alpha  Z] \\
&= I(c_\beta c_\alpha-s_\beta s_\alpha c_{2\theta}) -iX(s_\beta c_\alpha s_{2\theta})-iY(-s_\beta s_\alpha s_{2\theta}) -iZ(c_\beta s_\alpha + S_\beta c_\alpha c_{2\theta}).
\end{align}
Therefore, we have
\begin{align}
c_\phi &= c_\beta c_\alpha -s_\beta s_\alpha c_{2\theta}, \label{eq:G_alpha_angle}\\
s_\phi \vec{n} &= (s_\beta c_\alpha s_{2\theta},\, -s_\beta s_\alpha s_{2\theta},\, c_\beta s_\alpha +s_\beta c_\alpha c_{2\theta} ). \label{eq:G_alpha_axis}
\end{align}

To simplify notations, we order the trigonometric functions in each monomial by the subscript's order $(\alpha_1,\alpha_2,\beta_1,\beta_2)$. For example, $ccss$ stands for $c_{\alpha_1}c_{\alpha_2}s_{\beta_1}s_{\beta_2}$, and furthermore, $ccs_{1+2}$ stands for $c_{\alpha_1} c_{\alpha_2} s_{\beta_1+\beta_2}$.
With this convention in mind, we first calculate the rotation angle $\phi$ of $F'_{xr} = G'(\alpha_2,\beta_2) G'(\alpha_1,\beta_1)$ using Lemma~\ref{lem:Bloch}.(4) and Eqs.~\eqref{eq:G_alpha_angle}, \eqref{eq:G_alpha_axis} as follows.
\begin{align}
c_\phi &= (c_{\alpha_1} c_{\beta_1} -s_{\alpha_1} s_{\beta_1} c_{2\theta}) (c_{\alpha_2} c_{\beta_2} -s_{\alpha_2} s_{\beta_2} c_{2\theta}) \notag\\
&\quad -ccss\cdot s_{2\theta}^2 -ssss\cdot s_{2\theta}^2 -(s_{\alpha_1} c_{\beta_1} +c_{\alpha_1} s_{\beta_1} c_{2\theta}) (s_{\alpha_2} c_{\beta_2} +c_{\alpha_2} s_{\beta_2} c_{2\theta}) \\
&= cccc -c_{2\theta}(scsc+cscs) +c_{2\theta}^2ssss \notag\\
&\quad -ccss-ssss +c_{2\theta}^2(ccss+ssss) -(sscc+c_{2\theta}(cssc+sccs)+c_{2\theta}^2ccss) \\
&=(cccc-ccss-sscc-ssss) -(1-2\lambda)(scsc+cscs+cssc+sccs) +(1-2\lambda)^2 2ssss.
\end{align}
Regrouping terms by the oder of $\lambda$, we obtain coefficient of the constant term $\lambda^0=1$ as
\begin{align}
&ccc_{1+2} -ssc_{1+2} -scs_{1+2} -css_{1+2}\\
=& c_{1+2}c_{1+2} -s_{1+2}s_{1+2} \\
=& \cos\frac{\alpha_1+\alpha_2+\beta_1+\beta_2}{2}.
\end{align}
The coefficients of $\lambda$ and $\lambda^2$ are $2\times(s_{1+2}s_{1+2} -4ssss)$ and $8ssss$, respectively.
Therefore, for $F'_{xr,\alpha} = G'(\alpha,\beta_2) G'(\alpha,\beta_1)$, its rotation angle $\phi$ satisfies
\begin{align}\label{eq:fxr_alpha_phi}
c_\phi &= \cos(\alpha+\frac{\beta_1+\beta_2}{2}) +2\lambda(\sin\alpha \sin\frac{\beta_1+\beta_2}{2} -4\sin^2\frac{\alpha}{2} \sin\frac{\beta_1}{2} \sin\frac{\beta_2}{2} ) 
\notag\\ &\quad
+ 8\lambda^2 \sin^2\frac{\alpha}{2} \sin\frac{\beta_1}{2} \sin\frac{\beta_2}{2}.
\end{align}
And for $F'_{xr,\beta} = G'(\alpha_2,\beta) G'(\alpha_1,\beta)$, its rotation angle $\phi$ satisfies
\begin{align}\label{eq:fxr_beta_phi}
    c_\phi &= \cos(\frac{\alpha_1+\alpha_2}{2}+\beta) + 2\lambda (\sin\frac{\alpha_1+\alpha_2}{2} \sin\beta -4\sin\frac{\alpha_1}{2} \sin\frac{\alpha_2}{2} \sin^2\frac{\beta}{2}) 
    \notag \\ &\quad
    + 8\lambda^2 \sin\frac{\alpha_1}{2} \sin\frac{\alpha_2}{2} \sin^2\frac{\beta}{2}.
\end{align}

We now calculate the $x-$coordinates $n_x$ of $F'_{xr}$'s rotation axis $\vec{n}$ as follows
\begin{align}
s_\phi n_x &= c_2 n_{1x} +c_1 n_{2x} +n_{2y} n_{1z} -n_{1y} n_{2z} \\
&= (c_{\alpha_2} c_{\beta_2} -s_{\alpha_2} s_{\beta_2} c_{2\theta})   c_{\alpha_1} s_{\beta_1} s_{2\theta} 
 + (c_{\alpha_1} c_{\beta_1} -s_{\alpha_1} s_{\beta_1} c_{2\theta})   c_{\alpha_2} s_{\beta_2} s_{2\theta}
\notag \\ &\quad
-s_{\beta_2} s_{\alpha_2} s_{2\theta} (s_{\alpha_1} c_{\beta_1} +c_{\alpha_1} s_{\beta_1} c_{2\theta}) 
+s_{\beta_1} s_{\alpha_1} s_{2\theta} (s_{\alpha_2} c_{\beta_2} +c_{\alpha_2} s_{\beta_2} c_{2\theta}) \\
s_{2\theta}\times  &= ccsc -csss\cdot c_{2\theta} +cccs -scss\cdot c_{2\theta}
\notag \\ &\quad
-sscs -csss\cdot c_{2\theta} +sssc +scss\cdot c_{2\theta} \\
s_{2\theta}\times  &= c_{1-2}sc +c_{1+2}cs -c_{2\theta} 2csss \\
&\quad \text{or:}\ ccs_{1+2} -sss_{1-2} -c_{2\theta} 2csss.
\end{align}
The notation `$s_{2\theta}\times$' is actually on the right hand side of the equal sign, we put it on the left to introduce less brackets. Therefore, for $F'_{xr,\alpha}$, we have
\begin{equation}
s_\phi n_x = 2\sqrt{\lambda(1-\lambda)} \Big( \sin\frac{\beta_1}{2} \cos\frac{\beta_2}{2} +\cos\alpha \cos\frac{\beta_1}{2} \sin\frac{\beta_2}{2} -(1-2\lambda)\sin\alpha \sin\frac{\beta_1}{2} \sin\frac{\beta_2}{2} \Big). \label{eq:fxr_alpha_nx}
\end{equation}
And for $F'_{xr,\beta}$, we have
\begin{equation}
s_\phi n_x = 2\sqrt{\lambda(1-\lambda)} \Big(\cos\frac{\alpha_1}{2} \cos\frac{\alpha_2}{2} \sin\beta - 2(1-2\lambda) \cos\frac{\alpha_1}{2} \sin\frac{\alpha_2}{2} \sin^2\frac{\beta}{2} \Big). \label{eq:fxr_beta_nx}
\end{equation}

We now calculate the $y-$coordinates $n_y$ of $F'_{xr}$'s rotation axis $\vec{n}$ as follows
\begin{align}
s_\phi n_y &= c_2 n_{1y} +c_1 n_{2y} -n_{2x} n_{1z} +n_{1x} n_{2z} \\
&= -(c_{\alpha_2} c_{\beta_2} -s_{\alpha_2} s_{\beta_2} c_{2\theta}) s_{\beta_1} s_{\alpha_1} s_{2\theta}
   -(c_{\alpha_1} c_{\beta_1} -s_{\alpha_1} s_{\beta_1} c_{2\theta}) s_{\beta_2} s_{\alpha_2} s_{2\theta}
\notag \\ &\quad
-c_{\alpha_2} s_{\beta_2} s_{2\theta}  (s_{\alpha_1} c_{\beta_1} +c_{\alpha_1} s_{\beta_1} c_{2\theta})
+c_{\alpha_1} s_{\beta_1} s_{2\theta}  (s_{\alpha_2} c_{\beta_2} +c_{\alpha_2} s_{\beta_2} c_{2\theta}) \\
s_{2\theta}\times &=
-scsc +ssss\cdot c_{2\theta} -cscs +ssss\cdot c_{2\theta}
\notag \\ &\quad
-sccs -ccss\cdot c_{2\theta} +cssc +ccss\cdot c_{2\theta}\\
s_{2\theta}\times &=
-s_{1-2}sc -s_{1+2}cs +2c_{2\theta}ssss \\
&\quad \text{or:}\ -scs_{1+2} +css_{1-2} +2c_{2\theta}ssss.
\end{align}
Therefore, for $F'_{xr,\alpha}$, we have
\begin{equation}
s_\phi n_y = 2\sqrt{\lambda(1-\lambda)} \Big( -\sin\alpha \cos\frac{\beta_1}{2} \sin\frac{\beta_2}{2} + 2(1-2\lambda) \sin^2\frac{\alpha}{2} \sin\frac{\beta_1}{2} \sin\frac{\beta_2}{2} \Big). \label{eq:fxr_alpha_ny}
\end{equation}
And for $F'_{xr,\beta}$, we have
\begin{equation}
s_\phi n_y = 2\sqrt{\lambda(1-\lambda)} \Big( -\sin\frac{\alpha_1}{2} \cos\frac{\alpha_2}{2} \sin\beta + 2(1-2\lambda) \sin\frac{\alpha_1}{2} \sin\frac{\alpha_2}{2} \sin^2\frac{\beta}{2} \Big). \label{eq:fxr_beta_ny}
\end{equation}

Finally, we calculate the $z-$coordinates $n_z$ of $F'_{xr}$'s rotation axis $\vec{n}$ as follows
\begin{align}
s_\phi n_z &= c_2 n_{1z} +c_1 n_{2z} +n_{2x} n_{1y} -n_{1x} n_{2y} \\
&= (c_{\alpha_2} c_{\beta_2} -s_{\alpha_2} s_{\beta_2} c_{2\theta}) (s_{\alpha_1} c_{\beta_1} +c_{\alpha_1} s_{\beta_1} c_{2\theta})
  +(c_{\alpha_1} c_{\beta_1} -s_{\alpha_1} s_{\beta_1} c_{2\theta}) (s_{\alpha_2} c_{\beta_2} +c_{\alpha_2} s_{\beta_2} c_{2\theta})
\notag \\ &\quad
-c_{\alpha_2} s_{\beta_2} s_{2\theta} \cdot  s_{\beta_1} s_{\alpha_1} s_{2\theta}
+c_{\alpha_1} s_{\beta_1} s_{2\theta} \cdot  s_{\beta_2} s_{\alpha_2} s_{2\theta}.
\end{align}
The last line is equal to $-scss +csss +c_{2\theta}^2(scss -csss)$. If we regard $c_{2\theta}$ as the pivot, and consider regrouping the terms by $1,c_{2\theta},c_{2\theta}^2$, we have
\begin{align}
s_\phi n_z &= (sccc +cscc -scss +csss) +c_{2\theta} c_{1+2}s_{1+2} -2c_{2\theta}^2 csss.
\end{align}
Therefore, for $F'_{xr,\alpha}$, we have
\begin{equation}
s_\phi n_z = \sin\alpha \cos\frac{\beta_1}{2} \cos\frac{\beta_2}{2} +(1-2\lambda) \cos\alpha \sin\frac{\beta_1+\beta_2}{2} -(1-2\lambda)^2 \sin\alpha \sin\frac{\beta_1}{2} \sin\frac{\beta_2}{2}. \label{eq:fxr_alpha_nz}
\end{equation}
And for $F'_{xr,\beta}$, we have
\begin{align}
s_\phi n_z &= \sin\frac{\alpha_1}{2} \cos\frac{\alpha_2}{2} \cos\beta +\cos\frac{\alpha_1}{2} \sin\frac{\alpha_2}{2}
\notag \\ &\quad
+(1-2\lambda) \cos\frac{\alpha_1+\alpha_2}{2} \sin\beta  -2(1-2\lambda)^2 \cos\frac{\alpha_1}{2} \sin\frac{\alpha_2}{2} \sin^2\frac{\beta}{2}. \label{eq:fxr_beta_nz}
\end{align}

\section{Continuity of $\beta_2=f(\beta_1)$}\label{app:f_continu}
\begin{lemma} \label{lem:continuum_f}
When $\beta_1\in[-\pi,\pi]$, the function $\beta_2=f(\beta_1)$ obtained by solving Eq.~\eqref{im1}, or equivalently Eq.~\eqref{eq:intro_fxr_alpha_im}, as shown in Section~\ref{subsec:continuous} is inherently continuous.
\end{lemma}
\begin{proof}
We copy Eq.~\eqref{2f1} in Section~\ref{subsec:continuous} here for convenience:
\begin{equation}
    t_2 = -\frac{t_1\big[ (1-2\lambda) c_{2\alpha} + 2\lambda \big] + s_{2\alpha}}{c_{2\alpha} - t_1 s_{2\alpha} (1-2\lambda)}. \tag{2f1} \label{2f1.1}
\end{equation}
Note that Eq.~\eqref{eq:intro_fxr_alpha_im} is equivalent to:
\begin{equation}
    s_2 \big[ c_1 c_{2\alpha} - s_1 s_{2\alpha} (1-2\lambda) \big] = -s_1 c_2 \big[ (1-2\lambda) c_{2\alpha} + 2\lambda \big] - c_1 c_2 s_{2\alpha}. \tag{im 2} \label{im2}
\end{equation}
The potential `discontinuity point' of $f(\beta_1)$ is the solution of the equation $t_1=\frac{c_{2\alpha}}{s_{2\alpha}(1-2\lambda)}$. To avoid the denominator of its RHS become zero, the following first two cases need to be considered first.
\begin{enumerate}
    \item $\alpha=\pi$. Then $s_{2\alpha}=\sin\alpha=0, c_{2\alpha}=-1$, so Eq.~\eqref{2f1.1} becomes $t_2=-(1-4\lambda)t_1$, and we denote it by Eq.~(2f1-1). Eq.~\eqref{im2} now also simplifies to $s_2 c_1 = -(1-4\lambda) s_1 c_2$ (im2-1). When $\beta_1\in(-\pi,\pi)$, function $\beta_2=f(\beta_1)$ obtained by taking $\arctan$ of Eq.~(2f1-1) is continuous, so we only need to consider the `discontinuity point' $\beta_1=\pm\pi$ of Eq.~(2f1-1). We have $c_1=0, s_1=1$. Substituting them into Eq.~(im2-1), we have $0=(1-4\lambda)\cos(\beta_2/2)$, which has solutions $\beta_2=\pm \pi$. Moreover, $\beta_2=f(\beta_1)\to\pm\pi$ when $\beta_1\to \pm\pi$. Therefore $f$ is inherently continuous at $\beta_1=\pm\pi$.
    \item $\alpha\neq\pi, \lambda=\frac{1}{2}$. Then Eq.~\eqref{2f1.1} becomes $t_2 = -\frac{t_1 + \sin\alpha}{\cos\alpha}$ (2f1-2) and Eq.~\eqref{im2} becomes $s_2 c_1 \cos\alpha= - s_1 c_2 - c_1 s_2 \sin\alpha$ (im2-2). Note that the denominator of Eq.~(2f1-2) is $\cos\alpha$, so we will need to discuss whether $\alpha=\pi/2$: 
    \begin{enumerate}
        \item $\alpha=\pi/2$. Then Eq.~(im2-2) becomes $0=\sin(\frac{\beta_1+\beta_2}{2})$. Solving it produce $\beta_2=f(\beta_1)=-\beta_1+2l\pi$ for $l\in\mathbb{Z}$, so $f$ is continuous.
        \item $\alpha \neq \pi/2$. Consider the `discontinuity point' $\beta_1=\pm\pi$ of Eq.~(2f1-2). Substituting $\beta_1=\pm\pi$ into Eq.~(im2-2), we obtain $0=c_2$, so $\beta_2=\pm\pi$. Moreover, the function $\beta_2=f(\beta_1)$ obtained by taking $\arctan$ of Eq.~(2f1-2) will $\to\pm\pi$ when $\beta_1\to\pm\pi$. Therefore $f$ is inherently continuous at $\beta_1=\pm\pi$.
    \end{enumerate}
    \item $\alpha\neq\pi, \lambda\neq\frac{1}{2}$. We first consider the `discontinuity point' which cause the denominator of Eq.~\eqref{2f1.1} to become zero, then we have $t_1=\frac{c_{2\alpha}}{s_{2\alpha}(1-2\lambda)}$ (suppose by taking $\arctan$ of both sides, its solution is the `discontinuity point' $\beta_1'$), which is equivalent to $c_1 c_{2\alpha} - s_1 s_{2\alpha} (1-2\lambda)=0$. Substituting it into Eq.~\eqref{im2}, we have $0 = s_1 c_2 \big[ (1-2\lambda) c_{2\alpha} + 2\lambda \big] + c_1 c_2 s_{2\alpha}$, which has solution $\beta_2=\pm\pi$. Moreover, the function $\beta_2=f(\beta_1)$ obtained by taking $\arctan$ of Eq.~\eqref{2f1.1} will $\to\infty$ (so $\beta_2\to\pm\pi$) when $\beta_1 \to \beta_1'$. Therefore $f$ is inherently continuous at $\beta_1'$. 
    Next, we consider the `discontinuity point' $\beta_1=\pm\pi$ of Eq.~\eqref{2f1.1}. Now $c_1=0, s_1=1$. Substituting them into Eq.~\eqref{im2}, we have $ s_2 s_{2\alpha}(1-2\lambda) = c_2 \big[ (1-2\lambda) c_{2\alpha} + 2\lambda \big] $, which has solution $t_2' := \frac{(1-2\lambda) c_{2\alpha} + 2\lambda}{s_{2\alpha}(1-2\lambda)}$. Moreover, from Eq.~\eqref{2f1.1}, $t_2\to t_2'$ when $\beta_1\to\pi$, so $f$ is also continuous at $\beta_1=\pm\pi$.
\end{enumerate}
\end{proof}
\end{appendices}

\end{document}